\documentclass[10pt, twocolumn, journal]{IEEEtran}
\usepackage{comment}

\usepackage{enumerate}
\usepackage{graphicx}
\usepackage{epsf}
\usepackage{subfigure}
\usepackage{amsmath}
\usepackage{amssymb}
\usepackage{array}
\usepackage{setspace}
\usepackage[ruled,lined]{algorithm2e}
\usepackage[amsmath,thmmarks,amsthm]{ntheorem}
\usepackage{algorithmic}
\usepackage{color}
\usepackage{amsfonts}
\usepackage{dsfont}
\usepackage{xfrac}
\usepackage{breqn}
\usepackage{bbm}
\usepackage{cite}
\usepackage{makecell}
\usepackage{mathtools}
\usepackage{mathrsfs}
\usepackage{float}
\usepackage{adjustbox}

\DeclareMathOperator*{\argmax}{argmax}
\DeclareMathOperator*{\argmin}{argmin}

\usepackage{amsbsy}

\graphicspath{{fig/}}

\newtheorem{lemma}{Lemma}
\newtheorem{corollary}{Corollary}

\begin{document}
\title{\huge D-STAR: Dual Simultaneously Transmitting and Reflecting Reconfigurable Intelligent Surfaces for Joint Uplink/Downlink Transmission}

\author{
Li-Hsiang Shen,~\IEEEmembership{Member,~IEEE},
Po-Chen Wu,~\IEEEmembership{Student Member,~IEEE}, Chia-Jou Ku, 
Yu-Ting Li,~\IEEEmembership{Student Member,~IEEE},
Kai-Ten Feng,~\IEEEmembership{Senior Member,~IEEE},
Yuanwei Liu,~\IEEEmembership{Fellow,~IEEE},
Lajos Hanzo,~\IEEEmembership{Life Fellow,~IEEE}}

\maketitle

\begin{abstract}
	The joint uplink/downlink (JUD) design of simultaneously transmitting and reflecting reconfigurable intelligent surfaces (STAR-RIS) is conceived in support of both uplink (UL) and downlink (DL) users. Furthermore, the dual STAR-RISs (D-STAR) concept is conceived as a promising architecture for 360-degree full-plane service coverage, including UL/DL users located between the base station (BS) and the D-STAR as well as beyond. The corresponding regions are termed as primary (P) and secondary (S) regions. Both BS/users exist in the P-region, but only users are located in the S-region. The primary STAR-RIS (STAR-P) plays an important role in terms of tackling the P-region inter-user interference, the self-interference (SI) from the BS and from the reflective as well as refractive UL users imposed on the DL receiver. By contrast, the secondary STAR-RIS (STAR-S) aims for mitigating the S-region interferences. The non-linear and non-convex rate-maximization problem formulated is solved by alternating optimization amongst the decomposed convex sub-problems of the BS beamformer, and the D-STAR amplitude as well as phase shift configurations. We also propose a D-STAR based active beamforming and passive STAR-RIS amplitude/phase (DBAP) optimization scheme to solve the respective sub-problems by Lagrange dual with Dinkelbach's transformation, alternating direction method of multipliers (ADMM) with successive convex approximation (SCA), and penalty convex-concave procedure (PCCP). Our simulation results reveal that the proposed D-STAR architecture outperforms the conventional single RIS, single STAR-RIS, and half-duplex networks. The proposed DBAP of D-STAR outperforms the state-of-the-art solutions found in the open literature for different numbers of quantization levels, geographic deployment, transmit power and for diverse numbers of transmit antennas, patch partitions as well as D-STAR elements.
\end{abstract}


{\let\thefootnote\relax\footnotetext
{Li-Hsiang Shen is with the Department of Communication Engineering, National Central University, Taoyuan 320317, Taiwan. (email: shen@ncu.edu.tw)}}

{\let\thefootnote\relax\footnotetext
{Po-Chen Wu, Chia-Jou Ku, Yu-Ting Li, Kai-Ten Feng are with the Department of Electronics and Electrical Engineering, National Yang Ming Chiao Tung University (NYCU), Hsinchu 300093, Taiwan. (email: wupochen.ee11@nycu.edu.tw, amyku0722@gmail.com, tammy10306@gmail.com, and
ktfeng@nycu.edu.tw)}}

{\let\thefootnote\relax\footnotetext
{Yuanwei Liu is with the School of Electronic Engineering and Computer Science, Queen Mary University of London, E1 4NS London, U.K. (e-mail: yuanwei.liu@qmul.ac.uk).
}}

{\let\thefootnote\relax\footnotetext
{Lajos Hanzo is with Next Generation Wireless, University of Southampton, SO17 1BJ Southampton, U.K. (email: lh@ecs.soton.ac.uk)}}

\section{Introduction}\label{INT}

  Reconfigurable intelligent surfaces (RIS) made of meta-material are capable of beneficially ameliorating the wireless propagation environments \cite{acm,mag_exp}. This is achieved by appropriately configuring the phase shifts of its reflective elements with the aid of passive beamforming for circumventing non-line-of-sight (NLoS) propagation. However, an impediment of RISs is that the transmitter and users are geometrically confined within the same 180-degree half-plane, rather than supporting  roaming across the entire 360-degree area \cite{acm}. By contrast, the simultaneously transmitting and reflecting RIS (STAR-RIS) architecture is capable of circumventing this limitation \cite{theory_star}, extending the service to the full coverage area. Hence, it is also termed as an intelligent omni-surface (IOS) \cite{mag_exp, mag_model, mag_mode, star_cite1, star_cite2}. The first prototype based experiment was reported in \cite{mag_exp}, confirming the feasibility of the STAR function in practice.

\begin{table*}[!ht]
	\centering
	\scriptsize
	\caption {Comparison of Open Literature}
	\resizebox{\textwidth}{!}{
\begin{tabular}{|l||c|c|c|c|c|c|c|c|c|c|c|c|c|}
\hline
 & \cite{mag_mode, ana_noma1, ana_noma2, ana_noma3} & \cite{star_UL1} & \cite{star_UL2} & \cite{star_3d1} & \cite{star_3d2} & \cite{quant} & \cite{couple1} & \cite{couple2} & \cite{couple3} & \cite{ris_fd_my, fd} & \cite{fd_ok} & \cite{fd_no1, fd_no2, fd_no3} & This Work \\ \hline\hline
Transmission type & DL & UL & UL & DL & DL & DL & DL & DL & DL & JUD & JUD & JUD & JUD \\ \hline
STAR-RIS & \checkmark & \checkmark & \checkmark & \checkmark & \checkmark & \checkmark & \checkmark & \checkmark & \checkmark & RIS & \checkmark & \checkmark & \pmb{\checkmark} \\ \hline
\begin{tabular}[c]{@{}l@{}}(Service Coverage)\\ DL user in P-, S-region\\ UL user in P-, S-region\end{tabular} & \begin{tabular}[c]{@{}c@{}}\checkmark , \checkmark\\ - , -\end{tabular} & \begin{tabular}[c]{@{}c@{}}- , -\\ \checkmark , \checkmark\end{tabular} & \begin{tabular}[c]{@{}c@{}}- , -\\ \checkmark , \checkmark\end{tabular} & \begin{tabular}[c]{@{}c@{}}\checkmark , \checkmark\\ - , -\end{tabular} & \begin{tabular}[c]{@{}c@{}}\checkmark , \checkmark\\ - , -\end{tabular} & \begin{tabular}[c]{@{}c@{}}\checkmark , \checkmark\\ - , -\end{tabular} & \begin{tabular}[c]{@{}c@{}}\checkmark , \checkmark\\ - , -\end{tabular} & \begin{tabular}[c]{@{}c@{}}\checkmark , \checkmark\\ - , -\end{tabular} & \begin{tabular}[c]{@{}c@{}}\checkmark , \checkmark\\ - , -\end{tabular} & N/A & \begin{tabular}[c]{@{}c@{}}- , \checkmark\\ \checkmark , -\end{tabular} & \begin{tabular}[c]{@{}c@{}}\checkmark , -\\ - , \checkmark\end{tabular} & \begin{tabular}[c]{@{}c@{}}\pmb{\checkmark} , \pmb{\checkmark}\\ \pmb{\checkmark} , \pmb{\checkmark}\end{tabular} \\ \hline
Active beamforming &  & \checkmark & \checkmark & \checkmark & \checkmark &  & \checkmark & \checkmark & \checkmark & \checkmark &  &  & \pmb{\checkmark} \\ \hline
Passive beamforming & \checkmark & \checkmark & \checkmark & \checkmark & \checkmark & \checkmark & \checkmark & \checkmark & \checkmark & \checkmark & \checkmark & \checkmark & \pmb{\checkmark} \\ \hline
Rate guarantee &  & \checkmark &  &  & \checkmark & \checkmark &  & \checkmark & \checkmark & \checkmark & \checkmark & \checkmark & \pmb{\checkmark} \\ \hline
Coupled phase shifts &  &  & \checkmark &  &  &  & \checkmark & \checkmark & \checkmark &  &  &  & \pmb{\checkmark} \\ \hline
Power constraint &  & \checkmark & \checkmark & \checkmark & \checkmark &  & \checkmark & \checkmark & \checkmark & \checkmark &  &  & \pmb{\checkmark} \\ \hline
Quantization evaluation &  & \checkmark &  &  & \checkmark & \checkmark &  & \checkmark &  &  &  &  & \pmb{\checkmark} \\ \hline
Deployment evaluation &  &  & \checkmark &  &  &  &  &  &  & \checkmark &  &  & \pmb{\checkmark} \\ \hline
Multi-surfaces &  &  &  &  &  &  &  &  &  &  &  &  & \pmb{\checkmark} \\ \hline
General user distribution &  &  &  &  &  &  &  &  &  &  &  &  & \pmb{\checkmark} \\ \hline
\end{tabular}} \label{comparetable}
\end{table*}

    There exist three different operating protocols of STAR-RISs \cite{mag_mode}, namely the energy splitting (ES), mode selection (MS) and time-switching (TS) mechanisms. ES splits the element-wise energy between reflecting and transmitting the signals, whereas MS is regarded as a reflection-only or transmission-only assignment of the STAR-RIS elements. Finally, TS is operated by switching the elements between the reflection and transmission modes in a time-division manner. It was shown in \cite{mag_mode} that ES is the most beneficial mechanism of providing multicast and multiuser services. Half-duplex (HDx) communications aided by RIS/STAR-RIS is considered to support either uplink (UL) or downlink (DL) transmission in a time- or frequency-division manner. The authors of \cite{ana_noma1, ana_noma2, ana_noma3} have analyzed the theoretically attainable effective ergodic rate of a single STAR-RIS in the DL of a non-orthogonal multiple access scheme. The statistical characteristics of the channels are considered in \cite{ana_noma1}, whilst the closed-form expression of the rate achieved by the individual near and far users are derived in \cite{ana_noma2}. Moreover, the authors of \cite{ana_noma3} take into account the additional factor of real-time quality of service. In \cite{star_UL1, star_UL2}, the authors employ STAR-RIS in the UL for improving the secrecy rate and the spectral efficiency, respectively. By contrast, the authors of \cite{star_3d1, star_3d2} further leverage the STAR-RIS architecture in a three-dimensional scenario for robust transmissions. The STAR-RIS also has its own hardware limitations, with one of them owing to the quantization of its phase shifts \cite{quant}. The authors of \cite{couple1, couple2, couple3} additionally consider a practical coupled phase shifts based on the meta-material constraints detailed in \cite{theory_star}. A general STAR-RIS framework was firstly proposed in \cite{couple1} for determining the amplitude and phase shifts are firstly proposed in \cite{couple1}. In \cite{couple2}, a certain minimum secrecy capacity was guaranteed subject to the constraints of BS transmit power budget and STAR-RIS amplitude/phase-shift coupling. In \cite{couple3}, advanced machine learning methods were designed to conduct joint BS/STAR-RIS beamforming in support of multiuser services. Hybrid control is designed, along with high-dimensional continuous amplitude and discrete phase shifts.

   Nonetheless, the RIS/STAR-RIS relying on HDx potentially leads to 50$\%$ spectral erosion compared to full-duplex (FD). As a remedy, a joint uplink/downlink (JUD) regime is conceived for matching the throughput of FD systems \cite{ris_fd_my, fd, fd_ok, fd_no1, fd_no2, fd_no3}. Note that FD is more specific for an antenna supporting both UL/DL at the same time, while JUD separates the whole antenna set into DL transmitter and UL receiver antennas. In a typical FD network, the BS and users can be operated in FD mode. The most challenging problem in both JUD and FD is the complex nature of the interferences induced by the DL BS and UL user transmissions. However, it can be alleviated by exploiting advanced transmission techniques, such as non-orthogonal multiple access (NOMA) \cite{noma_fd} and rate-splitting multiple access (RSMA) \cite{rsma_fd}. Depending on the strong/weak channel quality, different UL/DL NOMA user groups can be formed in JUD, having superposed signals of various power levels. The terminology of RSMA in JUD implies that both the UL/DL streams can be cooperatively partitioned into a common and a private message segment. In both methods, sophisticated successive interference cancellation should be harnessed for extracting the desired user signals. Alternatively, RIS/STAR-RIS provides a simpler solution associated with a comparably high channel diversity for mitigating the UL/DL interferences. In our previous work \cite{ris_fd_my}, we considered JUD transmission using conventional RISs, which can only have half-plane service, as mentioned previously. In \cite{fd}, weighted sum rate maximization is considered in two-way communications with the BS and multiple users both in the FD mode. However, a single STAR-RIS with the single-directional STAR functionality is unable to support complete JUD transmission. In \cite{fd_ok}, the authors minimize the power consumption of JUD in STAR-RIS, while considering a pair of UL/DL users restricted to their reflection/transmission regions. In \cite{fd_no1}, STAR-RIS assisted JUD wireless communication is considered in order to maximize the weighted sum rate of the system. In \cite{fd_no2}, large-scale statistics of the channel state information are leveraged in a STAR-RIS assisted system supporting two users. The authors of \cite{fd_no3} provided a quantitative analysis of practical energy efficient two-way JUD communications assisted by the STAR-RIS. Although the authors of \cite{fd_no1, fd_no2, fd_no3} consider JUD while using a STAR-RIS, their models are incompatible with realistic imperfect electric circuits as well as with the theory of electromagnetism. Recently, the authors of \cite{dual_star} have proposed to adopt a bi-directional STAR-RIS architecture derived from \cite{theory_star}, which is capable of receiving the incident signals at both sides. The transfer function is provided in \cite{dual_star} with the transmission and reflection coefficients depending on the electric and magnetic impedance of metasurfaces. However, in dual-directional STAR-RIS, those coefficients could be the same, which might lead to more complex interference management. To elaborate a little further, by appropriately configuring both the RIS-based passive beamforming as well as the active beamforming at the base station (BS), the self-interference (SI) \cite{ris_fd_my, fd_ok} of JUD can be alleviated. 
   
   A table contrasting our contribution at a glance to the literature is provided in Table \ref{comparetable}. We can infer from Table \ref{comparetable} that most of the existing RIS/STAR-RIS solutions can only support half-plane coverage with a single optimized STAR-RIS in both HDx and JUD transmissions. The important impacts of geographic deployment and quantization are not evaluated in most of works. Motivated by the above-mentioned issues, we have conceived a new architecture termed as dual STAR-RISs (D-STAR)\textsuperscript{\ref{note1}}\footnotetext[1]{However, this dual STAR-RIS philosophy is different from that of the double-RIS concept in \cite{h1,h2}, which has one RIS near the BS and one near the users to be able to get around the blockage in the middle both in the UL and DL. \label{note1}}, which relies on a pair of STAR-RISs combined with 180-degree orientations. Explicitly, in D-STAR one of the reflective surfaces is facing toward the BS, while the other one is facing in the opposite direction. As a benefit, an exact 360-degree service provision can be achieved for JUD transmission, as it will be detailed in Section \ref{SM} with reference to Fig. \ref{fig:starall}. The main contributions of this paper can be summarized as follows.
\begin{itemize}
     \item To support full coverage in JUD transmission, the new D-STAR architecture operating in the ES mode is conceived. D-STAR separates the whole coverage into a P-region wherein the BS and users exist and S-region with only users. The primary STAR-RIS (STAR-P) of Fig. \ref{fig:starall} deals with the P-region inter-user interference, plus with the SI remanating from the BS and from the reflective as well as from the refractive UL users and contaminating the DL receiver. By contrast, the secondary STAR-RIS (STAR-S) aims for mitigating the S-region inter-user interferences.
     
     \item We consider the problem of DL throughput maximization, guaranteeing a specific UL rate requirement when optimizing the BS's active beamforming and the amplitudes/phase shifts of D-STAR. The original non-linear and non-convex problem is solved by alternating optimization after decomposing it into convex sub-problems. We propose a D-STAR based active beamforming and passive STAR-RIS amplitude/phase (DBAP) scheme. We then solve the respective weighting coefficient optimization sub-problems by the Lagrange dual based method in conjunction with Dinkelbach's transformation \cite{Dinkel}, the alternating direction method of multipliers (ADMM) \cite{admm, admm2}, the successive convex approximation (SCA) \cite{SCA} and the penalty convex-concave procedure (PCCP) \cite{pccp}.
     
     \item The proposed DBAP scheme relying on the D-STAR architecture is evaluated through simulations by taking into account the reflection coefficient, quantization effects, the inter-D-STAR distances, the transmit power as well as the number of splitting D-STARs, antennas and elements. Moreover, the D-STAR is compared to MS \cite{bm1}, to coupled phases \cite{couple2}, and to conventional STAR-RIS as well as to RISs under JUD/HDx transmissions. We will demonstrate that the proposed D-STAR achieves the highest rate amongst the existing methods found in the open literature.
     
\end{itemize}

The rest of the article is organized as follows. In Section \ref{SM}, we introduce the architecture, system model and problem formulation of D-STAR. In Section \ref{D-STAR}, we describe our proposed DBAP scheme relying on the D-STAR architecture considering both the active beamforming and the D-STAR configuration. Our performance evaluations of D-STAR are discussed in Section \ref{NR}. Finally, our conclusions are offered in Section \ref{CON}.

\section{System Model and Problem Formulation}\label{SM}

\subsection{D-STAR Architecture}

\begin{figure}[!t]
		\centering
		\subfigure[]
		{\includegraphics[width=1.3 in]{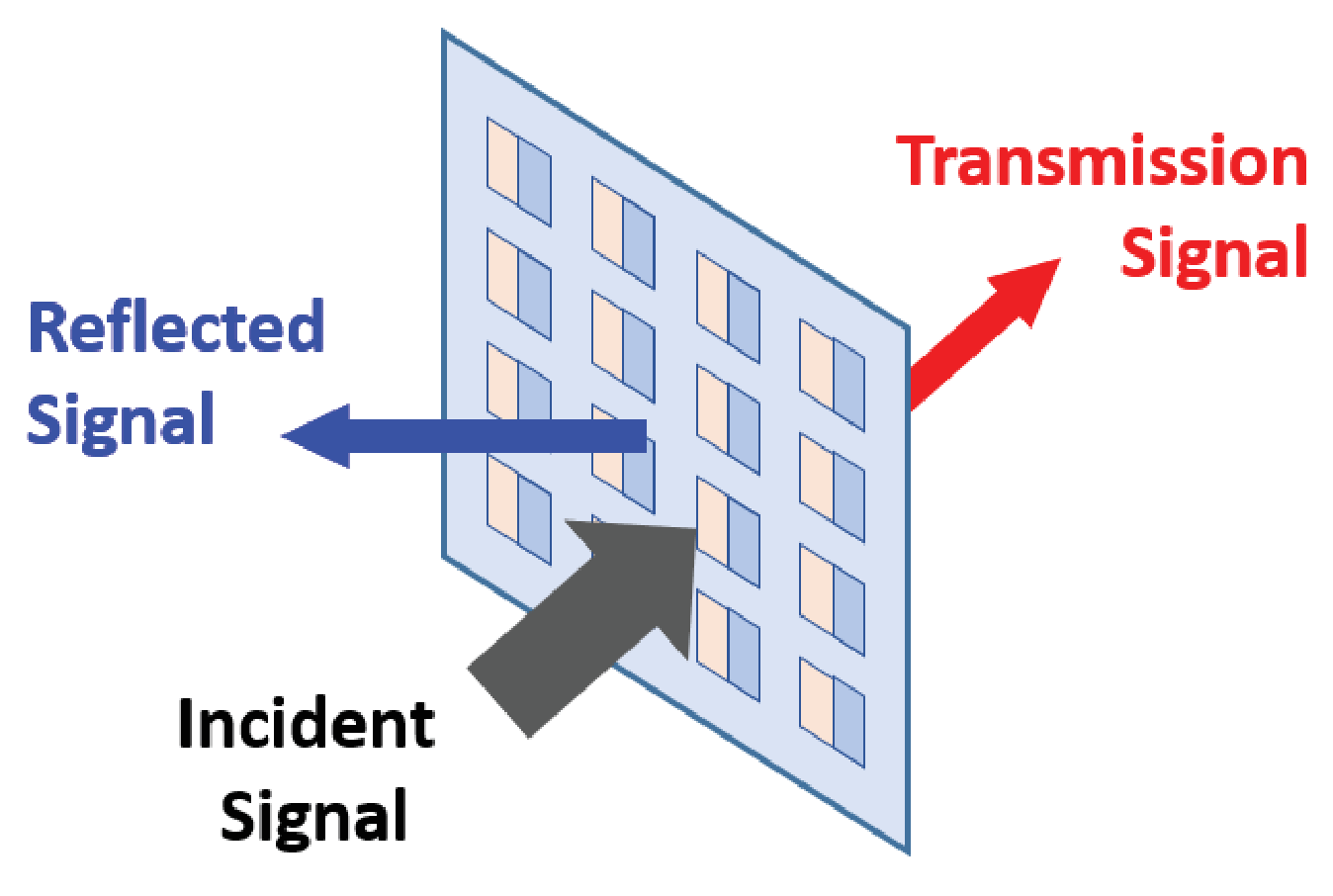} \label{fig:star}}
		\subfigure[]
		{\includegraphics[width=1.8 in]{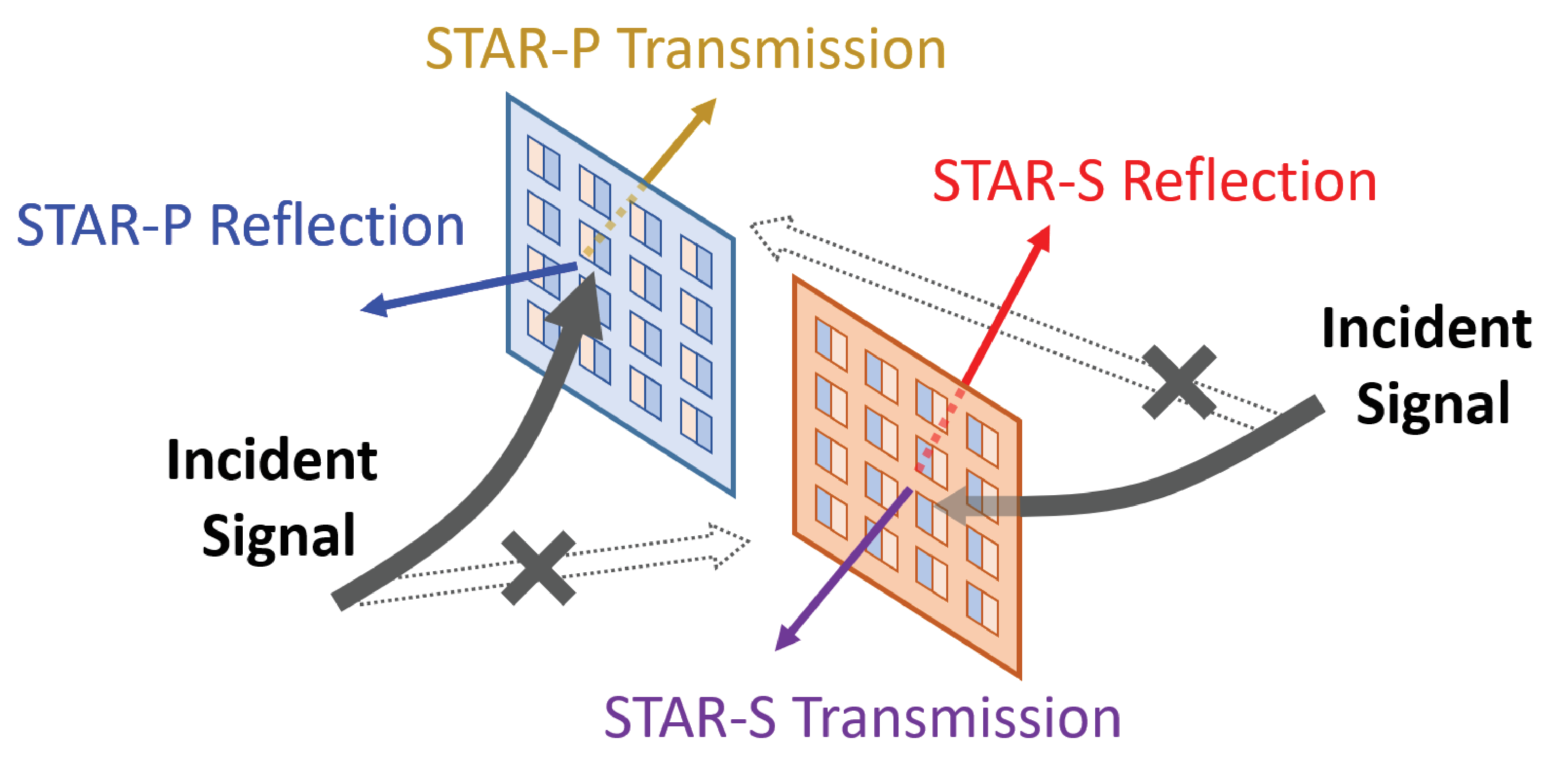} \label{fig:star2}}
		\caption{Architectures of STAR-RIS for (a) conventional STAR-RIS and (b) proposed D-STAR. }
		\label{fig:starall}
	\end{figure}

\begin{figure}[!t]
\centering
\includegraphics[width=3.3in]{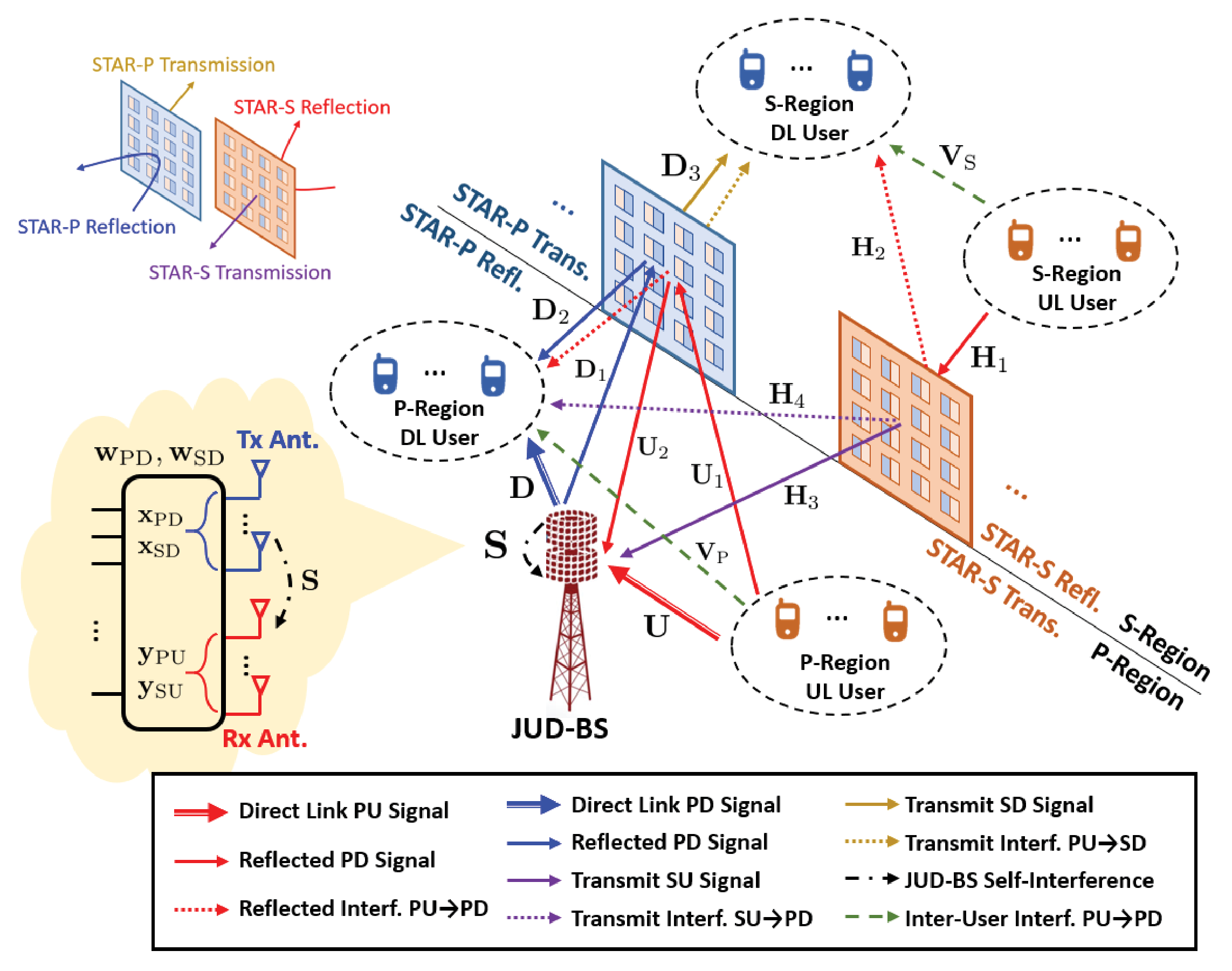}
 \caption{The overall architecture of proposed D-STAR system.} \label{Fig.1}
\end{figure}

\begin{figure*}[!t]
		\centering
		\subfigure[]
		{\includegraphics[width=1.4 in]{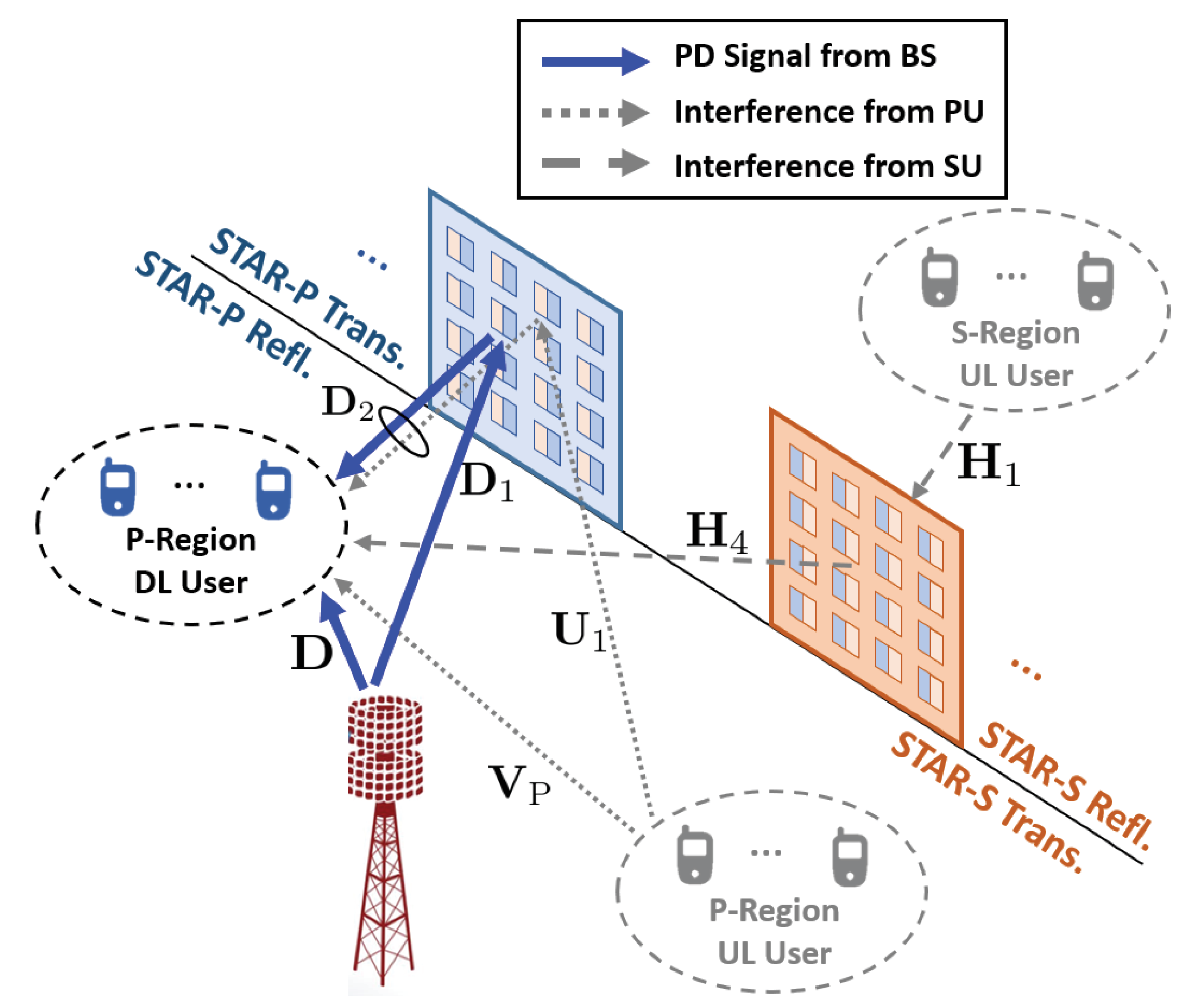} \label{fig:1}}
		\subfigure[]
		{\includegraphics[width=1.75 in]{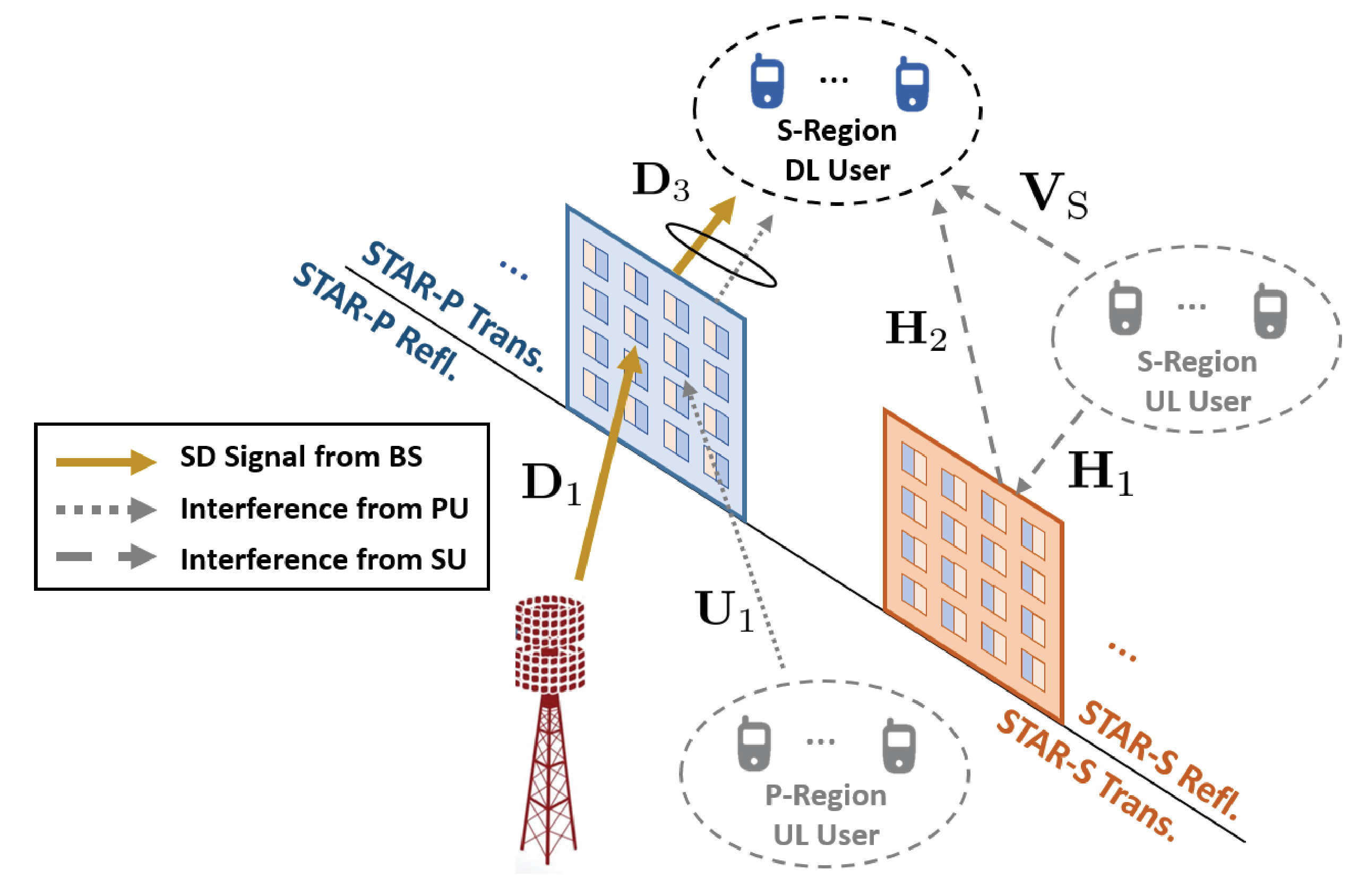} \label{fig:2}}
		\subfigure[]
		{\includegraphics[width=1.75 in]{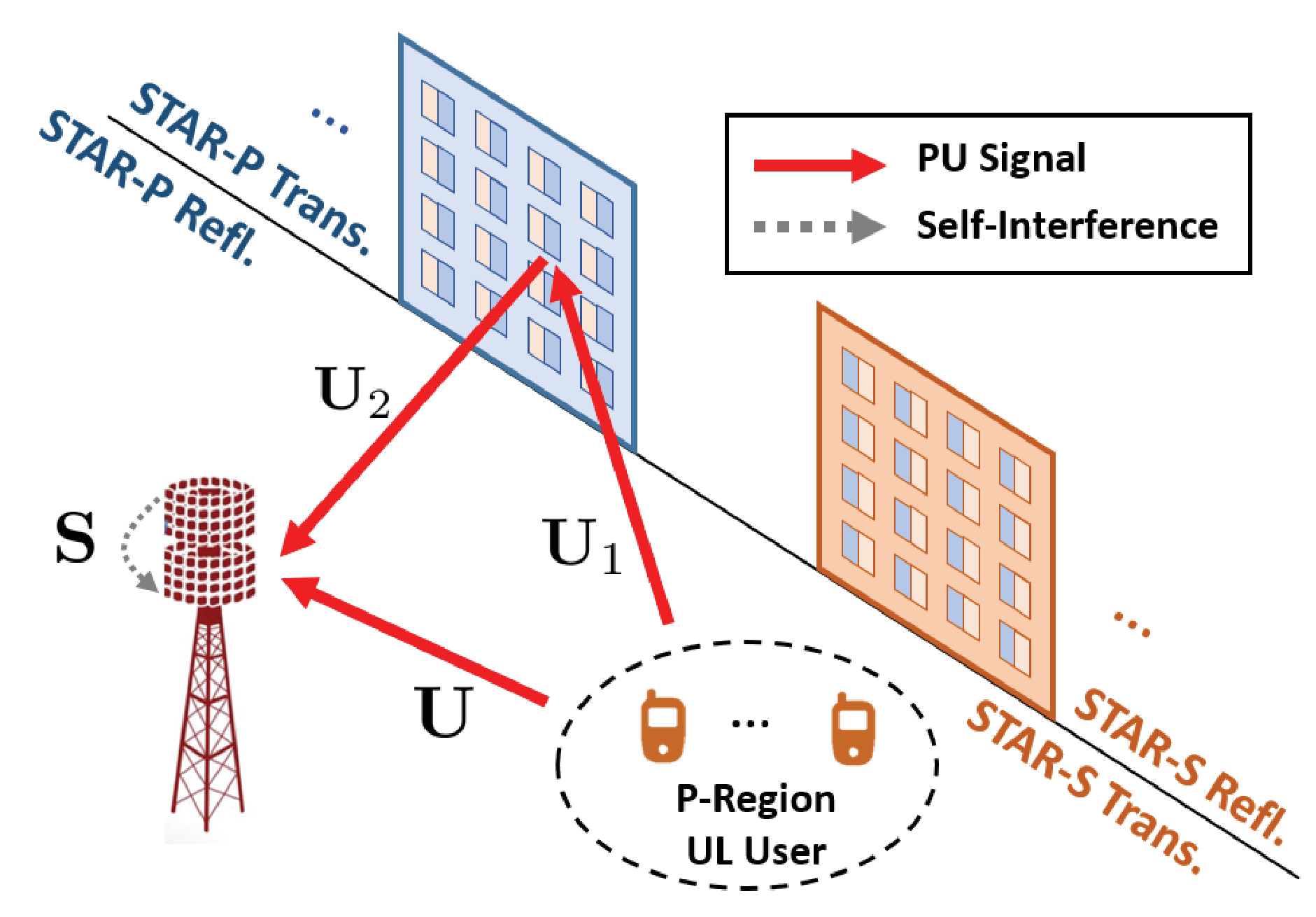} \label{fig:3}}
		\subfigure[]
		{\includegraphics[width=1.75 in]{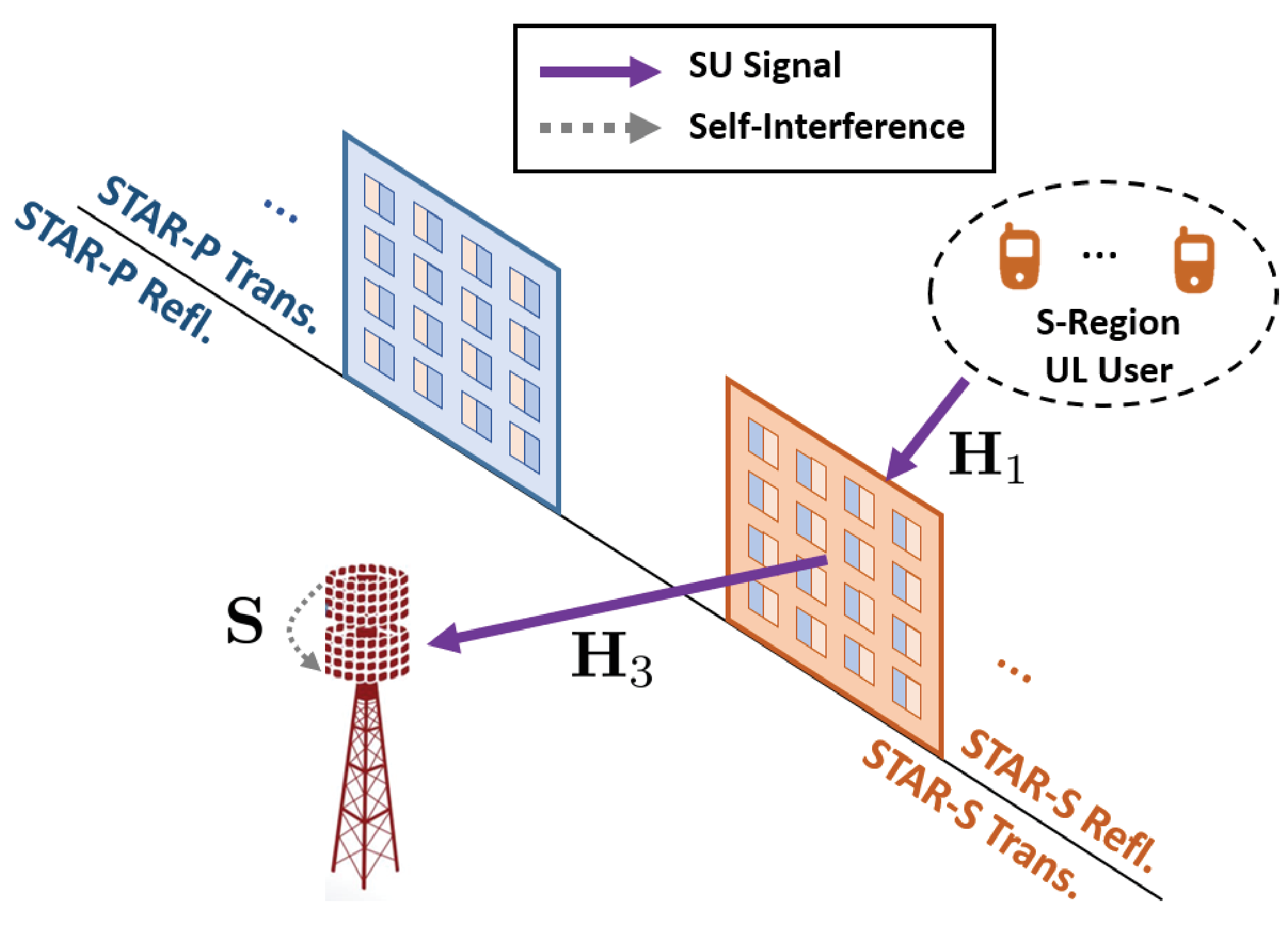} \label{fig:4}}
		\caption{Separate architectures of received signal paths for (a) primary DL, (b) secondary DL, (c) primary UL, and (d) secondary UL users.}
		\label{fig:sep}
	\end{figure*}

As depicted in Fig. \ref{fig:star}, the existing architecture of a single STAR-RIS relies on splitting the incident signal power into reflected and transmitted signals. Note that we consider two single-directional STAR-RISs, which means that the impinging signals may only arrive from one side of the STAR-RIS. This architecture is capable of supporting a DL service, if we consider the incident signal in Fig. \ref{fig:star} as the DL signal arriving from the BS. Once UL users exist, the STAR-RIS should be rotated by 180-degrees for supporting UL services, which may result in an impractical deployment scenario. Therefore, as shown in Fig. \ref{fig:star2}, we have conceived a novel D-STAR system consisting of two STAR-RISs, separating the whole coverage area into primary and secondary regions, termed as P-region and S-region, respectively. The P-region includes both the BS and users, whereas only users exist in the S-region. The STAR-P indicates that the reflection-side is facing towards the BS, whilst the STAR-S is operated with the transmission-side facing towards the BS. Observe from Fig. \ref{fig:star2} that the signal incident from the left can only illuminate STAR-P, generating reflective and refractive signal powers, but it cannot be directed towards the STAR-S. Similarly, STAR-S can only reflect and transmit the signal incident from the right. As a result, a complete 360-degree coverage area can be provided for all incoming signals. Note that there may exist an imperfect alignment for a pair of STAR-RISs with smaller than 180-degree shift, which potentially generates different coverage types. This case will be more focused on the coverage issues which are beyond the scope of this paper. We define the four respective STAR-RISs as $\mathcal{M}_{x}=\{1,2,...,M_x\}$, where $x\in\mathcal{X}=\{\text{PT},\text{PR},\text{ST},\text{SR}\}$, where PT denotes STAR-P transmission and PR represents STAR-P reflection, while ST/SR denote STAR-S transmission and reflection, respectively. Note that we have $M_{\text{PT}} = M_{\text{PR}}$ and $M_{\text{ST}} = M_{\text{SR}}$ due to having identical surfaces with the same number of elements. The phase shift is denoted as
\begin{align}
\boldsymbol{\Theta}_{x} = {\rm diag}(\boldsymbol{\phi}_x) = {\rm diag}({\beta_{x,1}} \theta_{x,1},...,{\beta_{x,M_x}} \theta_{x,M_x}),
\end{align}
where we have $\theta_{x,m} = e^{j \vartheta_{x,m}}$ along with $0<\vartheta_{x,m} \leq 2\pi, \forall m\in \mathcal{M}_x$. The absolute value of phase is constrained by one, i.e., 
\begin{align}
|\theta_{x,m}| = 1. \label{Con_T}
\end{align}
If coupled phase shifts are considered, we will have additional constraints given by \cite{couple1, couple2, couple3}
\begin{align}
	\cos(\vartheta_{\text{PT},m} - \vartheta_{\text{PR},m})=0,
	\quad
	\cos(\vartheta_{\text{ST},m} - \vartheta_{\text{SR},m})=0. \label{Con_T1}
\end{align}
Since each element of D-STAR shares the same total power, the magnitude is confined to
\begin{align}
\beta_{\text{PT},m}^2+\beta_{\text{PR},m}^2=1,
\quad \beta_{\text{ST},m}^2+\beta_{\text{SR},m}^2=1. \label{Con_B}
\end{align}
We have our candidate solution set of D-STAR as $\boldsymbol{\beta}=\{ {\boldsymbol{\beta}}_{\text{PR}} , {\boldsymbol{\beta}}_{\text{PT}}, {\boldsymbol{\beta}}_{\text{SR}},{\boldsymbol{\beta}}_{\text{ST}} \}$ for the amplitudes and $\boldsymbol{\theta}=\{ {\boldsymbol{\theta}}_{\text{PR}} , {\boldsymbol{\theta}}_{\text{PT}}, {\boldsymbol{\theta}}_{\text{SR}},{\boldsymbol{\theta}}_{\text{ST}} \}$ for the phase shifts.

\begin{table*}[t]
\linespread{1.1}
\centering
\footnotesize
  \caption{Symbol Definition of D-STAR}
 \begin{tabular}{lll}
  \hline
  \textbf{Symbol} & \textbf{Dimension}  & \textbf{Definition} \\
  \hline
  $\mathbf{D}=[ \mathbf{d}_{1}, ...,\mathbf{d}_{K_{\text{PD}}} ] $ 
  & $ \mathbb{C}^{N_T \times K_{\text{PD}}}$ 
  & Direct link from BS to PD user
   \\
   
  $\mathbf{D}_1 $ 
  & $ \mathbb{C}^{M_{\text{PR}} \times N_T}$ 	
  & Reflected channel from BS to STAR-P           \\
  
  $\mathbf{D}_2 = [ \mathbf{d}_{2,1}, ...,\mathbf{d}_{2,K_{\text{PD}}} ] $ 
  & $ \mathbb{C}^{M_{\text{PR}} \times K_{\text{PD}}}$ 	
  & Reflected channel from STAR-P to PD user      \\
  
  $\mathbf{D}_3=[ \mathbf{d}_{3,1}, ...,\mathbf{d}_{3,K_{\text{SD}}} ] $ 
  & $ \mathbb{C}^{M_{\text{PT}} \times N_T}$ 	
  & Transmit channel from STAR-P to SD user       \\
  
  $\mathbf{U}=[ \mathbf{u}_{1}, ...,\mathbf{u}_{K_{\text{PU}}} ]$  	
  & $ \mathbb{C}^{N_R \times K_{\text{PU}}}$ 
  &  Direct link from PU user to BS  \\
  
  $\mathbf{U}_1=[ \mathbf{u}_{1,1}, ...,\mathbf{u}_{1,K_{\text{PU}}} ]$ 
  &  $ \mathbb{C}^{M_{\text{PR}} \times K_{\text{PU}}}$ 
  &  Reflected channel from PU to STAR-P \\
  
  $\mathbf{U}_2$ 	
  & $ \mathbb{C}^{N_R \times M_{\text{PR}}}$  
  & Reflected channel from STAR-P to BS \\
  
  $\mathbf{H}_1=[ \mathbf{h}_{1}, ...,\mathbf{h}_{K_{\text{SU}}} ]$  	
  & $ \mathbb{C}^{M_{\text{PR}} \times K_{\text{SU}}}$  
  & Reflected channel from SU user to STAR-S \\
  
  $\mathbf{H}_2=[ \mathbf{h}_{1}, ...,\mathbf{h}_{K_{\text{SD}}} ]$ 	
  &            $ \mathbb{C}^{M_{\text{SR}} \times K_{\text{SD}}}$  
  & Reflected channel from STAR-S to SD user \\
  
  $\mathbf{H}_3$ 
  & $ \mathbb{C}^{ N_R \times M_{\text{ST}} }$  
  & Transmit channel from STAR-S to BS \\
  
  $\mathbf{H}_4=[ \mathbf{h}_{1}, ...,\mathbf{h}_{K_{\text{PD}}} ]$ 	
  &  $ \mathbb{C}^{ M_{\text{ST}} \times K_{\text{PD}} }$  
  &  Transmit channel from STAR-S to PD user    \\
  
  $\mathbf{S}$ 	&  $ \mathbb{C}^{ N_R \times N_T }$  
  &  Self-interference channel   \\
  
  $\mathbf{V}_{\text{P}}=[ \mathbf{v}_{\text{P},1}, ...,\mathbf{v}_{\text{P},K_{\text{PD}}} ]$ 
  &        $ \mathbb{C}^{ K_{\text{PU}} \times K_{\text{PD}} }$  
  & Interfered channel from PU user to PD user \\

  $\mathbf{V}_{\text{S}}=[ \mathbf{v}_{\text{P},1}, ...,\mathbf{v}_{\text{P},K_{\text{SD}}} ]$ 
  &        $ \mathbb{C}^{K_{\text{SU}} \times K_{\text{SD}} }$  
  &  Interfered channel from SU user to SD user \\
  
  $\mathbf{w}_{x}$ 
  & $\mathbb{C}^{N_T}$ 
  & BS transmit beamforming vector      \\
  
  $\boldsymbol{\Theta}_{x}$ 
  & $\mathbb{C}^{M_{x} \times M_{x}}$ 
  & Diagonal matrix of D-STAR configuration\\
  
  $\boldsymbol{\theta}_{x} $ 
  & $ \mathbb{C}^{M_x}$ 
  & Amplitude vector of D-STAR      \\
  
  $\boldsymbol{\beta}_{x}$ 
  & $ \mathbb{R}^{M_x}$ 
  &  Phase shift vector of D-STAR          \\
  \hline
 \end{tabular} \label{channelDSTAR}
\end{table*}

\subsection{SINR Model} 
The overall architecture of the proposed D-STAR is depicted in Fig. \ref{Fig.1}. We consider a single JUD BS having $N_T$ transmit antennas and $N_R$ receiving antennas for DL and UL services, respectively. Note that the UL/DL antennas at the BS are operated independently, i.e., a single antenna element cannot transmit and receive signals simultaneously. It can be readily observed that all users are separated by two STAR-RISs into the region closer to and farther away from the BS, which are termed again as primary and secondary users, respectively. We assume that a total of $K = K_{\text{PD}} + K_{\text{SD}} + K_{\text{PU}} + K_{\text{SU}}$ users are uniformly distributed in the P- and S-region, which are categorized into $K_{\text{PD}}$, $K_{\text{SD}}$ primary/secondary DL (PD/SD) users, and $K_{\text{PU}}$, $K_{\text{SU}}$ primary/secondary UL (PU/SU) users. The user set is denoted by $\mathcal{K}_{u}=\{1,2,...,K_u \}$, where $u\in\mathcal{U}=\{\text{PD}, \text{SD}, \text{PU}, \text{SU}\}$ is the user index for the respective regional users of $\text{PD}$, $\text{SD}$, $\text{PU}$, and $\text{SU}$. We assume that the users in the S-region cannot receive signals directly from the BS due to the highly attenuated or blocked signal paths. For better analyzing the received signal model, we partition the overall network architecture of Fig. \ref{Fig.1} into four regions, as seen in Figs. \ref{fig:1} to \ref{fig:4} with all notations defined in Table \ref{channelDSTAR}. The respective received signal models are elaborated on as follows.
\begin{itemize}
	\item Primary DL users (\textbf{PD}): In Fig. \ref{fig:1}, the PD users can receive their DL signals directly from the BS via the channel $\mathbf{D}$, whilst the reflected signal arrives from the STAR-P reflection via the cascaded channel of $\mathbf{D}_1$ from the BS to STAR-P and via $\mathbf{D}_2$ from STAR-P to the PD users. Since by definition all DL transmit antennas are used for the DL signals, there exist intra-DL and inter-DL user group interferences. Moreover, additional interferences are imposed by the UL users, including PU and SU users. The PU users induce two types of interferences, namely direct link interference associated with channel $\mathbf{V}_{\text{P}}$ and reflected one with interference arriving via the cascaded channel of $\mathbf{U}_1$ from PU to STAR-P and $\mathbf{D}_2$ from STAR-P to PD users. As for SU users, it only imposes refractive interference at STAR-S via the channel $\mathbf{H}_1$ from SU to STAR-S and $\mathbf{H}_4$ from STAR-S to PD users. Therefore, we can obtain the overall received signal model of PD user $k$ as
	\begingroup
	\allowdisplaybreaks
	\begin{align}
	y_{\text{PD},k} &= 
	 \underbrace{\left( {\rm\mathbf{d}}_{2,k}^{\mathcal{H}} {\boldsymbol{\Theta}}_{\text{PR}}{\mathbf{D}}_{1}+{\mathbf{d}}_{k}^{\mathcal{H}} \right) \mathbf{w}_{\text{PD},k} {x}_{\text{PD},k}}_{\text{ PD Signal}} \notag\\
   & + \sum_{k' \in \mathcal{K}_{\text{PD}} \backslash k} \underbrace{\left( {\rm\mathbf{d}}_{2,k}^{\mathcal{H}} {\boldsymbol{\Theta}}_{\text{PR}}{\mathbf{D}}_{1}+{\mathbf{d}}_{k}^{\mathcal{H}} \right) \mathbf{w}_{\text{PD},k'} {x}_{\text{PD},k'}}_{\text{PD Interference}} \notag\\
   & +  \underbrace{\left( {\rm\mathbf{d}}_{2,k}^{\mathcal{H}} {\boldsymbol{\Theta}}_{\text{PR}}{\mathbf{D}}_{1}+{\mathbf{d}}_{k}^{\mathcal{H}} \right) \mathbf{x}_{\text{SD}}}_{\text{SD Interference}} \notag \\ 
   & + \sum_{k'\in \mathcal{K}_{\text{PU}}} \underbrace{\left( {\rm\mathbf{d}}_{2,k}^{\mathcal{H}} {\boldsymbol{\Theta}}_{\text{PR}} {\rm\mathbf{u}}_{1,k'} + {v}_{\text{P},k',k} \right) {x}_{\text{PU},k'}}_{\text{PU Interference}} \notag\\
   & + \sum_{k'\in \mathcal{K}_{\text{SU}}} \underbrace{{\mathbf{h}}_{4,k}^{\mathcal{H}} {\boldsymbol{\Theta}}_{\text{ST}} {\rm\mathbf{h}}_{1,k'} {x}_{\text{SU},k'}}_{\text{SU Interference}}
   + {n}_{\text{PD},k}, \label{y_nd}
	\end{align}	
\endgroup
	where
$\mathbf{x}_{\text{PD}} = \sum_{k\in \mathcal{K}_{\text{PD}}} \mathbf{w}_{\text{PD},k} {x}_{\text{PD},k}$ and 
$\mathbf{x}_{\text{SD}} = \sum_{k\in \mathcal{K}_{\text{SD}}} \mathbf{w}_{\text{SD},k} {x}_{\text{SD},k}$ are the beamformed signals associated with the transmit beamforming vectors defined as $\mathbf{w}_{\text{PD},k}$ and $\mathbf{w}_{\text{SD},k}$ for PD and SD users, respectively.
	
	\item Secondary DL users (\textbf{SD}): In Fig. \ref{fig:2}, SD users can only receive their DL signals from the transmission side of STAR-P via the cascaded channel of $\mathbf{D}_1$ from BS to STAR-P and $\mathbf{D}_3$ from STAR-P to SD users. Likewise, the DL signals are subject to intra- and inter-DL user group interferences. Moreover, additional interferences arrive from PU and SU users. The PU users inflict refractive interferences at STAR-P via the channel $\mathbf{U}_1$ from PU to STAR-P and $\mathbf{D}_3$ from STAR-P to SD users. Two types of interferences impinge from the SU users, i.e., the direct link interference via channel $\mathbf{V}_{\text{S}}$ and the reflected one via the cascaded channel of $\mathbf{H}_1$ from SU to STAR-S and $\mathbf{H}_2$ from STAR-S to SD users. Therefore, we can formulate the overall received signal model of SD user $k$ as
		\begingroup
	\allowdisplaybreaks
	\begin{align}
	y_{\text{SD},k} &= 
  	\underbrace{ \mathbf{d}_{3,k}^{\mathcal{H}} \mathbf{\Theta}_{\text{PT}} \mathbf{D}_{1} \mathbf{w}_{\text{SD},k} {x}_{\text{SD},k}}_{\text{SD Signal}}\notag\\
   & + \sum_{k' \in \mathcal{K}_{\text{SD}}} \underbrace{\mathbf{d}_{3,k}^{\mathcal{H}} \mathbf{\Theta}_{\text{PT}} \mathbf{D}_{1} \mathbf{w}_{\text{SD},k'} {x}_{\text{SD},k'}}_{\text{SD Interference}} \notag\\
   & +  \underbrace{ \mathbf{d}_{3,k}^{\mathcal{H}} \mathbf{\Theta}_{\text{PT}} \mathbf{D}_{1}  \mathbf{x}_{\text{PD}}}_{\text{PD Interference}} 
   + \sum_{k'\in \mathcal{K}_{\text{PU}}} \underbrace{\mathbf{d}_{3,k}^{\mathcal{H}} \mathbf{\Theta}_{\text{PT}} \mathbf{u}_{1,k'} {x}_{\text{PU},k'} }_{\text{PU Interference}} \notag\\
   & + \sum_{k'\in \mathcal{K}_{\text{SU}}} \underbrace{ \left( \mathbf{h}_{2,k}^{\mathcal{H}} \mathbf{\Theta}_{\text{SR}} \mathbf{h}_{1,k'} +{v}_{\text{S},k',k} \right) {x}_{\text{SU},k'} }_{\text{SU Interference}}
 	+ {n}_{\text{SD},k}. \label{y_fd}
	\end{align}
	\endgroup
	
	\item Primary UL users (\textbf{PU}): In Fig. \ref{fig:3}, the BS is capable of receiving the primary uplink user signals from two paths, including the directly link $\mathbf{U}$ from PU to BS as well as the reflected link through STAR-P reflection from $\mathbf{U}_1$ (PU to STAR-P) and $\mathbf{U}_2$ (STAR-P to BS). However, SI is imposed by the downlink transmit signals of both PD/SD users via the channel $\mathbf{S}$. Therefore, we can express the received signal model for PU users at the BS as
	\begin{align}
	\mathbf{y}_{\text{PU}} &= 
 \sum_{k\in \mathcal{K}_{\text{PU}}} \left(  {\rm\mathbf{U}}_2 {\boldsymbol{\Theta}}_{\text{PR}} {\rm\mathbf{u}}_{1,k} + \rm\mathbf{u}_{k} \right) {x}_{\text{PU},k} 
\notag\\
   & + \left( {\mathbf{S}}+{\mathbf{U}}_2 {\boldsymbol{\Theta}}_{\text{PR}} {\rm\mathbf{D}}_1 \right) \left( \mathbf{x}_{\text{PD}} + \mathbf{x}_{\text{SD}}  \right)
+ \textbf{n}_{\text{PU}}. \label{y_nu}
	\end{align}
	
	\item Secondary UL users (\textbf{SU}): In Fig. \ref{fig:4}, the BS can only receive the S-region uplink signals from the transmission region of STAR-S via the cascaded channel of $\mathbf{H}_1$ (SU to STAR-S) and $\mathbf{H}_3$ (STAR-S to BS). Again, SI is imposed by the downlink transmit signals of both PD/SD users via the channel $\mathbf{S}$. Therefore, we can formulate the received signal model of SU users at the BS as
	\begin{align}
	\mathbf{y}_{\text{SU}} &= 
  \sum_{k\in \mathcal{K}_{\text{SU}}} \mathbf{H}_{3} \mathbf{\Theta}_{\text{ST}} \mathbf{h}_{1,k} {x}_{\text{SU},k} 
\notag\\
   & + \left( {\mathbf{S}}+{\mathbf{U}}_2 {\boldsymbol{\Theta}}_{\text{PR}} {\rm\mathbf{D}}_1 \right) \left( \mathbf{x}_{\text{PD}} + \mathbf{x}_{\text{SD}}  \right)
+ \textbf{n}_{\text{SU}}. \label{y_fu}
	\end{align}
\end{itemize}
Note that ${n}_{u,k}$ and $\textbf{n}_{u}$ are defined as noise. Moreover, we can observe from $\eqref{y_nu}$ and $\eqref{y_fu}$ that perfect self-interference cancellation cannot be conducted because in most cases ${\boldsymbol{\Theta}}_{\text{PR}}$ should strike a compromise between the desired signal alignment and self-interference alleviation. Therefore, based on $\eqref{y_nd}$ to $\eqref{y_fu}$, we can obtain the signal-to-interference-plus-noise ratio (SINR) for the users respectively as
\begingroup
 \allowdisplaybreaks
\begin{align}
	& \gamma_{\text{PD},k} = \notag \\
	& \frac{ \lVert {\rm\mathbf{d}}_{\text{PD},k} \mathbf{w}_{\text{PD},k} \rVert^2}{ \smashoperator[r]{\sum\limits_{k' \in \mathcal{K}_{\text{PD}} \backslash k}}\lVert {\rm\mathbf{d}}_{\text{PD},k} {\rm\mathbf{w}}_{\text{PD},k'} \rVert^2 \!+\! \smashoperator[r]{\sum\limits_{k' \in \mathcal{K}_{\text{SD}}}} \lVert {\rm\mathbf{d}}_{\text{PD},k} {\rm\mathbf{w}}_{\text{SD},k'} \rVert^2 + w_{\text{PD},k} + \sigma^{2}}, \label{sinr_nd}\\
	& \gamma_{\text{SD},k} = \notag\\
	& \frac{\lVert {\rm\mathbf{d}}_{\text{SD},k} \mathbf{w}_{\text{SD},k}  \rVert^2}{ \smashoperator[r]{\sum\limits_{k' \in \mathcal{K}_{\text{SD}} \backslash k}} \lVert {\rm\mathbf{d}}_{\text{SD},k} {\rm\mathbf{w}}_{\text{SD},k'} \rVert^2 + \smashoperator[r]{\sum\limits_{k' \in \mathcal{K}_{\text{PD}}}} \lVert {\rm\mathbf{d}}_{\text{SD},k} {\rm\mathbf{w}}_{\text{PD},k'} \rVert^2 + w_{\text{SD},k} + \sigma^{2}}, \label{sinr_fd}\\
	& \gamma_{\text{PU},k} = \frac{\lVert  {\rm\mathbf{U}}_2 {\boldsymbol{\Theta}}_{\text{PR}} {\rm\mathbf{u}}_{1,k} + \rm\mathbf{u}_{k} \lVert^2}{ \sum_{k' \in \mathcal{K}_{\text{PD}}} \lVert \mathbf{S}_t {\rm\mathbf{w}}_{\text{PD},k'} \rVert^2 + \sum_{k' \in \mathcal{K}_{\text{SD}}} \lVert {\rm\mathbf{S}}_{t} {\rm\mathbf{w}}_{\text{SD},k'} \rVert^2 + \sigma^2}, \label{sinr_nu} \\
	& \gamma_{\text{SU},k} = \frac{\lVert {\rm\mathbf{H}}_3 {\boldsymbol{\Theta}}_{\text{ST}} {\rm\mathbf{h}}_{1,k} \lVert^2 }{ \sum_{k' \in \mathcal{K}_{\text{PD}}} \lVert \mathbf{S}_t {\rm\mathbf{w}}_{\text{PD},k'} \rVert^2 + \sum_{k' \in \mathcal{K}_{\text{SD}}} \lVert {\rm\mathbf{S}}_{t} {\rm\mathbf{w}}_{\text{SD},k'} \rVert^2 + \sigma^2}, \label{sinr_fu}
\end{align}
\endgroup
where we define notations of
\begingroup
 \allowdisplaybreaks
\begin{align*}
    {\rm\mathbf{d}}_{\text{PD},k} &= {\rm\mathbf{d}}_{2,k}^{\mathcal{H}} {\boldsymbol{\Theta}}_{\text{PR}}{\mathbf{D}}_{1} + {\mathbf{d}}_{k}^{\mathcal{H}},
    \quad {\mathbf{d}}_{\text{SD},k} = {\rm\mathbf{d}}_{3,k}^{\mathcal{H}} {\boldsymbol{\Theta}}_{\text{PT}} {\mathbf{D}}_{1}, \notag \\
    \quad \mathbf{S}_t &= {\rm\mathbf{S}}+{\rm\mathbf{U}}_2 {\boldsymbol{\Theta}}_{\text{PR}} {\rm\mathbf{D}}_1,  \nonumber \\
    w_{\text{PD},k} &= \lVert \left( {\rm\mathbf{d}}_{2,k}^{\mathcal{H}} {\boldsymbol{\Theta}}_{\text{PR}} {\rm\mathbf{U}}_1 + {\rm\mathbf{v}}_{\text{P},k}^{\mathcal{H}} \right) {\rm\mathbf{x}}_{\text{PU}} \lVert^2 
    + \lVert {\rm\mathbf{h}}_{4,k}^{\mathcal{H}} {\boldsymbol{\Theta}}_{\text{ST}} {\rm\mathbf{H}}_1 {\rm\mathbf{x}}_{\text{SU}}\lVert^2  \nonumber \\
    w_{\text{SD},k} &= \lVert  {\rm\mathbf{d}}_{3,k}^{\mathcal{H}} {\boldsymbol{\Theta}}_{\text{PT}} {\mathbf{U}}_{1} {\mathbf{x}}_{\text{PU}} \lVert^2
    + \lVert \left( {\rm\mathbf{h}}_{2,k}^{\mathcal{H}} {\boldsymbol{\Theta}}_{\text{SR}}{\rm\mathbf{H}}_1 + {\rm\mathbf{v}}_{\text{S},k}^{\mathcal{H}} \right) {\mathbf{x}}_{\text{SU}} \lVert^2.  \nonumber
\end{align*} 
\endgroup 
Note that the total UL interferences expressed in vectorial form for the DL users are identical. Since the interferences are more detrimental than the noise, we consider the same noise power as $\sigma^2$. The individual ergodic rate can be expressed as
\begin{align}
	R_{u} = \sum_{k \in \mathcal{K}_u} \log_2 \left(1+\gamma_{u,k} \right).
\end{align}

\subsection{Problem Formulation}

In our D-STAR-enabled JUD network, we formulate our designed problem as maximizing the effective DL throughput, while guaranteeing the UL rate requirement, which is given by
\begingroup
 \allowdisplaybreaks
 \begin{subequations} \label{pro_origin}
  \begin{align} \mathop{\max}\limits_{\substack{{\mathbf{w}}_{\text{PD}},{\mathbf{w}}_{\text{SD}}, \\ \boldsymbol{\Theta} = \{\boldsymbol{\beta}, \boldsymbol{\theta} \}}} &\quad R_{\text{PD}} + R_{\text{SD}} \label{obj}
     \\
     \text{s.t.} \quad& \eqref{Con_T},\eqref{Con_T1}, \eqref{Con_B}, \label{C0} \\
     &	R_{\text{PU}} \geq R_{\text{PU},th}, \label{C1} \\
     &  R_{\text{SU}} \geq R_{\text{SU},th},\label{C2} \\
     & \sum_{k\in\mathcal{K}_{\text{PD}}} \mathbf{w}_{\text{PD},k}^{\mathcal{H}} \mathbf{w}_{\text{PD},k} + \sum_{k\in\mathcal{K}_{\text{SD}}} \mathbf{w}_{\text{SD},k}^{\mathcal{H}} \mathbf{w}_{\text{SD},k} \leq P_{t} \label{C3}.
\end{align}
 \end{subequations}
  \endgroup 
The constraint $\eqref{C0}$ represents the D-STAR configuration limitation. The constraints $\eqref{C1}$ and $\eqref{C2}$ respectively guarantee that the UL rates of the PU and SU are higher than the predefined thresholds of $R_{\text{PU},th}$ and $R_{\text{SU},th}$. The constraint of $\eqref{C3}$ represents the maximum allowable power of $P_t$. It is worth mentioning that the whole optimization process is carried out at the BS side. In this work, we focus more on the active BS beamforming as well as passive STAR-RIS configuration optimization at DL side, rather than UL power control. Accordingly, we set equal UL power for all UL users. We can observe that problem $\eqref{pro_origin}$ is complex due to its non-linearity and non-convexity as well as owing to the joint optimization of continuous variables and the discretized selection of the D-STAR phases. Therefore, we harness the DBAP scheme for solving the above-mentioned problems in the following section.

 \section{Proposed DBAP Scheme in D-STAR Architecture}\label{D-STAR}

Given the complex problem in $\eqref{pro_origin}$ formulated for the active beamforming $\{\mathbf{w}_{\text{PD}}, \mathbf{w}_{\text{SD}} \}$ and D-STAR configuration $\Theta= \{\boldsymbol{\beta}, \boldsymbol{\theta} \} $, we employ alternating optimization (AO) by harnessing the divide-and-conquer philosophy by decomposing it into the sub-problems of beamforming and D-STAR configuration. We first introduce some useful lemmas that will be employed for solving our problem, i.e., the Lagrangian dual transform \cite{LDtransform}, the Dinkelbach's transformation \cite{Dinkel}, SCA \cite{SCA}, and ADMM \cite{admm, admm2}.

\begin{lemma} \label{lemma_LD}
(Modified Lagrangian Dual Transform): The original problem is equivalent to the transformed one associated with the auxiliary variable $\gamma_{u,k}$ for each ratio term in the SINR as
\begin{align} \label{LD}
	R_u \!=\! \smashoperator[r]{\sum_{k\in\mathcal{K}_u}} \log_2 (1+\gamma_{u,k}) \!-\! \smashoperator[r]{\sum_{k\in\mathcal{K}_u}} \gamma_{u,k} \!+\! \smashoperator[r]{\sum_{k\in\mathcal{K}_u}} \frac{(1+\gamma_{u,k}) A_k(\boldsymbol{\Xi}) }{A_k(\boldsymbol{\Xi}) \!+\! B_k(\boldsymbol{\Xi})},
\end{align}
where $A_k(\boldsymbol{\Xi})$ and $B_k(\boldsymbol{\Xi})$ represent the nominator and denominator terms of the SINR $\gamma_{u,k}$, respectively. Note that $\boldsymbol{\Xi}=\{{\mathbf{w}}_{\text{PD}},{\mathbf{w}}_{\text{SD}}, \boldsymbol{\Theta}\}$ represents the total candidate solution set.
\end{lemma} 
\begin{proof}
See Appendix \ref{Appendix_LD}.
\end{proof}

\begin{lemma} \label{lemma_DIN}
(Dinkelbach's Optimization): The fractional form in the third term in $\eqref{LD}$ can be transformed into an affine function as
\begin{align}\label{Din}
	\sum_{k\in\mathcal{K}_u} (1+\gamma_{u,k}) \left[ A_k(\boldsymbol{\Xi}) - \lambda_k \Big( A_k(\boldsymbol{\Xi})+ B_k(\boldsymbol{\Xi}) \Big) \right].
\end{align}
The original objective of $\eqref{obj}$ is alternately solved by employing $\eqref{Din}$. Note that the terms in $\sum_{k\in\mathcal{K}_u} \log_2 (1+\gamma_{u,k}) - \sum_{k\in\mathcal{K}_u} \gamma_{u,k}$ in $\eqref{LD}$ are constants acquired at the previous iteration, which can be neglected in the transformed objective.
\end{lemma}
\begin{proof}
See Appendix \ref{Appendix_DIN}.
\end{proof}

\begin{lemma} \label{lemma_Jen}
For a logarithmic function, we have $\sum_{k} \log(1+x_k) \geq \log(1+\sum_{k} x_k)$ with the arbitrary real variable satisfying $x_k\geq 0$.
\end{lemma}

\begin{lemma} \label{lemma_sca}
(SCA Procedure): Consider a function $f(x)$ partitioned into a function having concave plus convex terms as $f(x) = f^+(x) + f^-(x)$. To make $f(x)$ concave, we alternatively solve a lower bounded objective $f_{cav}(x)$, i.e., $f(x)\geq f_{cav}(x) = f^+(x) + f^-(x_0) + \nabla_x^{\mathcal{H}} f^-(x_0)(x-x_0)$, where $x_0$ is an arbitrary variable. Note that we have $f(x)\approx f_{cav}(x)$ when $x$ is sufficiently small.
\end{lemma}
\begin{proof}
See Appendix \ref{Appendix_sca}.
\end{proof}

\subsection{Optimization of Active Transmit Beamformer}

Based on Lemma \ref{lemma_DIN}, we define
\begingroup
\allowdisplaybreaks
\begin{align}
	& A_{k}(\mathbf{w}_{u,k}) = \lVert {\rm\mathbf{d}}_{u, k} {\rm\mathbf{w}}_{u,k} \lVert^2, \\
	& B_{k}(\mathbf{w}_{u,k}) =  \sum_{k'\in \mathcal{K}_{u} \backslash k} \lVert {\rm\mathbf{d}}_{u, k} {\rm\mathbf{w}}_{u,k'} \lVert^2 + \sum_{k' \in \mathcal{K}_{u'}} \lVert {\rm\mathbf{d}}_{u,k} {\rm\mathbf{w}}_{u',k'} \rVert^2 \notag \\
	& \qquad\qquad + w_{u,k} + \sigma^2,\\
	& C_{k}(\mathbf{w}_{u,k}) =  A_{k}(\mathbf{w}_{u,k}) + B_{k}(\mathbf{w}_{u,k}) = \notag \\
	& \smashoperator[r]{\sum_{k'\in \mathcal{K}_{u}}} \lVert {\rm\mathbf{d}}_{u, k} {\rm\mathbf{w}}_{u,k'} \lVert^2 + \smashoperator[r]{\sum_{k' \in \mathcal{K}_{u'}}} \lVert {\rm\mathbf{d}}_{u,k} {\rm\mathbf{w}}_{u',k'} \rVert^2 + w_{u,k} + \sigma^2,
\end{align}
\endgroup
where the tuple is denoted as $(u,u') \in \{ (\text{PD},\text{SD}), ( \text{SD},\text{PD})\}$. We define the auxiliary variables of $\gamma_{u,k} = \frac{A_k(\mathbf{w}_{u,k}^{(a)})}{B_k(\mathbf{w}_{u,k}^{(a)})}$ and $\lambda_{u,k} = \frac{A_k(\mathbf{w}_{u,k}^{(a)})}{C_k(\mathbf{w}_{u,k}^{(a)})}$, which are the solutions obtained at the previous iteration $a$. Based on Lemma \ref{lemma_Jen}, we have a lower bounded rate formulated as
\begin{align} \label{LB_UL}
	R_u \geq \log_2 \left(1 +  \sum_{k \in \mathcal{K}_u} \gamma_{u,k} \right).
\end{align}
Based on $\eqref{LB_UL}$, we can then have the alternative problem associated with the lower bounded constraints of UL rates of $\xi_u$ as follows:
 \begingroup
 \allowdisplaybreaks
 \begin{subequations} \label{problem_w}
\begin{align}	\mathop{\max}\limits_{{\rm\mathbf{w}}_{\text{PD}},{\rm\mathbf{w}}_{\text{SD}}} &\quad 	
 	\sum_{ \substack{u \in\{\text{PD}, \text{SD}\}, \\ k\in\mathcal{K}_{u}}} \left( 1+\gamma_{u,k} \right) \left[ A_{k}(\mathbf{w}_{u,k}) - \lambda_{u,k} C_{k}(\mathbf{w}_{u,k}) \right] \label{obj1} \\
     \text{s.t.} &\quad \eqref{C3}, \notag \\
     & \quad \xi_{u} -t_{u} \Big( \sum_{ \substack{u'\in\{\text{PD},\text{SD}\} , \\ k\in \mathcal{K}_{u'} } } \lVert {\rm\mathbf{S}}_t  {\rm\mathbf{w}}_{u',k} \rVert^2 + \sigma^2 \Big) \geq 0, \notag \\
     & \qquad\qquad\qquad\qquad\qquad\qquad \forall u\in\{ \text{PU}, \text{SU}\}, \label{new_C1} 
\end{align} 
  \end{subequations}
  \endgroup
where we have $t_u = 2^{R_{u,th}}-1,
    \xi_{\text{PU}} = \lVert \left( {\rm\mathbf{U}}+ {\rm\mathbf{U}}_2 {\boldsymbol{\Theta}}_{\text{PR}} {\rm\mathbf{U}}_1 \right) {\rm\mathbf{x}}_{\text{PU}} \lVert^2$, and $
    \xi_{\text{SU}}  = \lVert  {\rm\mathbf{H}}_3 {\boldsymbol{\Theta}}_{\text{ST}} {\rm\mathbf{H}}_1  {\rm\mathbf{x}}_{\text{SU}} \lVert^2.$
We can observe that $\eqref{obj1}$ represents a form of summation associated with a convex and a concave term, whereas $\eqref{new_C1}$ is non-convex. Therefore, we harness the SCA procedure associated with a first-order Taylor series \cite{SCA}, i.e., $f(x) = f(x_0) + \nabla_x^{\mathcal{H}} f(x_0)(x-x_0)$ where $x_0$ is an arbitrary variable. By defining $\tilde{A}_{k}(\mathbf{w}_{u,k}) = 2\mathfrak{R}\left\lbrace
( \mathbf{w}^{(a)}_{u,k})^{\mathcal{H}} \mathbf{d}_{u,k}^{\mathcal{H}} \right\rbrace \left(\mathbf{d}_{u,k} \mathbf{w}_{u,k} - \mathbf{d}_{u,k} \mathbf{w}^{(a)}_{u,k}\right)$, we can then have the following problem represented by
 \begingroup
 \allowdisplaybreaks
 \begin{subequations} \label{problem_w2}
\begin{align}	\mathop{\max}\limits_{{\rm\mathbf{w}}_{\text{PD}},{\rm\mathbf{w}}_{\text{SD}}} &\quad 
	\sum_{ \substack{u \in\{\text{PD}, \text{SD}\},\\ k\in\mathcal{K}_{u}}} \left( 1+\gamma_{u,k} \right) \left[ \tilde{A}_{k}(\mathbf{w}_{u,k}) - \lambda_{u,k} C_{k}(\mathbf{w}_{u,k}) \right]\\
     \text{s.t.} &\qquad  \eqref{new_C1}. \nonumber
\end{align}
 \end{subequations}
  \endgroup
We can observe that problem $\eqref{problem_w2}$ is convex and can be solved to obtain the optimum of $\mathbf{w}_{\text{PD}}$ and $\mathbf{w}_{\text{SD}}$.

\subsection{Optimization of D-STAR}
After obtaining the beamforming policy from $\eqref{problem_w2}$, we proceed to optimize $\boldsymbol{\Theta}$ in D-STAR. We define $\boldsymbol{\phi}_{x}$ as a vector form of $\boldsymbol{\Theta}_{x}$ based on the following lemma and collorary.

\begin{lemma} \label{lemma_sym}
 The expression having the coupled terms of the reflected channel $\mathbf{d}\in \mathbb{C}^{N_1 \times 1}$ and $\mathbf{D}\in \mathbb{C}^{N_1 \times N_2}$, D-STAR configuration $\boldsymbol{\Theta} = {\rm diag}(\boldsymbol{\phi}) \in \mathbb{C}^{N_1 \times N_1}$ and the beamforming $\mathbf{w} \in \mathbb{C}^{N_2 \times 1} $ is equivalent to the following expression:
\begin{align} \label{dd}
 	{\rm\mathbf{d}}^{\mathcal{H}} {\boldsymbol{\Theta}}{\mathbf{D}} \mathbf{w} = \mathbf{w}^{\rm{T}}
{\rm rep} ({\rm\mathbf{d}}^{\mathcal{H}},N_2,1) \odot {\mathbf{D}}^{\rm{T}} \boldsymbol{\phi},
 \end{align}
where $\odot$ is the Hadamard product, whilst ${\rm rep} ({\rm\mathbf{d}}^{\mathcal{H}},N_2,1)$ stacks repeated vectors ${\rm\mathbf{d}}^{\mathcal{H}}$ into a matrix associated with dimension $N_2\times 1$.
\end{lemma}
\begin{proof}
See Appendix \ref{Appendix_sym}.
\end{proof}
 
\begin{corollary} \label{Cor_dd}
From Lemma \ref{lemma_sym}, we may arrive at a more detailed form of ${\rm\mathbf{U}} {\boldsymbol{\Theta}}{\mathbf{D}} \mathbf{w}$, where ${\rm\mathbf{U}} \in \mathbb{C}^{N_3 \times N_1}$, namely
\begin{align} \label{DD}
	{\rm\mathbf{U}} {\boldsymbol{\Theta}}{\mathbf{D}} \mathbf{w}
	= \left( \mathbf{I}_{N_3} \otimes \mathbf{w}^{\rm{T}} \right) 
{\rm diag}(\left[\mathbf{U}_1, ..., \mathbf{U}_{N_3}\right] )
\odot  \left( \mathbf{I}_{N_3} \otimes \mathbf{D}^{\rm{T}} \right) \boldsymbol{\phi},
\end{align}
where $\otimes$ is Kronecker product, $\mathbf{I}_{N_3}$ is an identity matrix with a dimension of $N_3\times N_3$, and $\mathbf{U}_{n} = {\rm rep} (\mathbf{U}_{(n,:)},N_2,1), \forall 1\leq n\leq N_3$, where $\mathbf{U}_{(n,:)}$ indicates the $n$-th row vector of $\mathbf{U}$.
\end{corollary}
\begin{proof}
See Appendix \ref{Appendix_dd}.
\end{proof}

Therefore, similar to problem $\eqref{problem_w}$, the alternative problem of $\eqref{pro_origin}$ for designing our D-STAR configuration is given by
 \begingroup
 \allowdisplaybreaks
 \begin{subequations} \label{problem_T}
\begin{align}
    \mathop{\max}\limits_{\substack{{\boldsymbol{\phi}_{\text{PR}}},{\boldsymbol{\phi}_{\text{PT}}},\\{\boldsymbol{\phi}_{\text{SR}}},{\boldsymbol{\phi}_{\text{ST}}}}} &\ 
    \sum_{\substack{k\in \mathcal{K}_u,\\ (u,x,x')\in\mathcal{T}}} \left( 1+\gamma_{u,k} \right) \left[ A_{u,k}(\boldsymbol{\phi}_{x}) - \lambda_{u,k} C_{u,k}(\boldsymbol{\phi}_{x}, \boldsymbol{\phi}_{x'}) \right] \label{obj_D} \\
     \text{s.t.} \qquad& \eqref{C0}, \nonumber \\
     & \lVert \boldsymbol{\Psi}_{\text{PU}} {\boldsymbol{\phi}}_{\text{PR}} + {\rm\mathbf{u}}_{\text{PU}} \rVert^2 
     - t_{\text{PU}} \, \xi\left({\boldsymbol{\phi}}_{\text{PR}}\right) \geq 0, \label{C_T1} \\
    & \lVert \boldsymbol{\Psi}_{\text{SU}} {\boldsymbol{\phi}}_{\text{ST}} \rVert^2 
    - t_{\text{SU}} \, \xi({\boldsymbol{\phi}}_{\text{PR}}) \geq 0, \label{C_T2}
\end{align}
 \end{subequations}
  \endgroup
  where
 \begingroup
 \allowdisplaybreaks 
  \begin{subequations} 
  \begin{align}
   	& A_{\text{PD},k}(\boldsymbol{\phi}_{\text{PR}}) = \lVert \boldsymbol{\varphi}_{1,k,k}{\boldsymbol{\phi}}_{\text{PR}} + d_{1,k,k}  \rVert^2 , \\
   	& A_{\text{SD},k}(\boldsymbol{\phi}_{\text{PT}}) = \lVert \boldsymbol{\psi}_{1,k,k}{\boldsymbol{\phi}}_{\text{PT}} \rVert^2, \\
    & C_{\text{PD},k}(\boldsymbol{\phi}_{\text{PR}}, \boldsymbol{\phi}_{\text{ST}})  = \sum_{k'\in \mathcal{K}_{\text{PD}}} \lVert \boldsymbol{\varphi}_{1,k,k'}{\boldsymbol{\phi}}_{\text{PR}} + d_{1,k,k'} \lVert^2 \notag\\
    & + \sum_{k' \in \mathcal{K}_{\text{SD}}} \lVert \boldsymbol{\varphi}_{2,k,k'}{\boldsymbol{\phi}}_{\text{PR}} + d_{2,k,k'} \rVert^2 + \lVert \boldsymbol{\varphi}_{3,k} {\boldsymbol{\phi}}_{\text{PR}} + v_{\text{P},k} \rVert^2 \notag \\
    & +\lVert \boldsymbol{\varphi}_{4,k} {\boldsymbol{\phi}}_{\text{ST}} \rVert^2 + \sigma^2, \\  
    & C_{\text{SD},k}(\boldsymbol{\phi}_{\text{PT}}, \boldsymbol{\phi}_{\text{SR}}) \!=\! \smashoperator[r]{\sum_{k'\in \mathcal{K}_{\text{SD}}}} \lVert \boldsymbol{\psi}_{1,k,k'}{\boldsymbol{\phi}}_{\text{PT}} \lVert^2 
    + \smashoperator[r]{\sum_{k' \in \mathcal{K}_{\text{PD}}}} \lVert \boldsymbol{\psi}_{2,k,k'}{\boldsymbol{\phi}}_{\text{PT}} \rVert^2 \nonumber \\
    & + \lVert  \boldsymbol{\psi}_{3,k} {\boldsymbol{\phi}}_{\text{PT}} \rVert^2 + \lVert  \boldsymbol{\psi}_{4,k} {\boldsymbol{\phi}}_{\text{SR}} + v_{\text{S},k} \rVert^2 + \sigma^2, \\
   & \xi\left({\boldsymbol{\phi}}_{\text{PR}}\right) = \sum_{ \substack{u'\in\{\text{PD},\text{SD}\} , k\in \mathcal{K}_{u'} } } \lVert \boldsymbol{\Psi}_{k} {\boldsymbol{\phi}}_{\text{PR}} + \mathbf{s}_{u',k} \rVert^2 + \sigma^2,
 \end{align}
  \end{subequations}
  \endgroup 
where the set obeys $\mathcal{T}= \{ (\text{PD}, \text{PR}, \text{ST}), (\text{SD}, \text{PT}, \text{SR}) \}$ in $\eqref{obj_D}$, whilst $(u,u') \in \{ (\text{PD},\text{SD}),  ( \text{SD},\text{PD})\}$. The other symbols based on Lemma \ref{lemma_sym} and Collorary \ref{Cor_dd} are listed as
 \begingroup
 \allowdisplaybreaks
\begin{align*}
    & \boldsymbol{\varphi}_{1,k,k'}{\boldsymbol{\phi}}_{\text{PR}}  =  {\rm\mathbf{d}}_{2,k}^{\mathcal{H}} {\boldsymbol{\Theta}}_{\text{PR}} {\rm\mathbf{D}}_1 {\rm\mathbf{w}}_{\text{PD},k'}, \notag \\
    & \boldsymbol{\varphi}_{2,k,k'}{\boldsymbol{\phi}}_{\text{PR}}  =  {\rm\mathbf{d}}_{2,k}^{\mathcal{H}} {\boldsymbol{\Theta}}_{\text{PR}} {\rm\mathbf{D}}_1 {\rm\mathbf{w}}_{\text{SD},k'}, \
    \boldsymbol{\varphi}_{3,k} {\boldsymbol{\phi}}_{\text{PR}} =  {\rm\mathbf{d}}_{2,k}^{\mathcal{H}} {\boldsymbol{\Theta}}_{\text{PR}} {\rm\mathbf{U}}_1 {\rm\mathbf{x}}_{\text{PU}}, \\
    & \boldsymbol{\varphi}_{4,k} {\boldsymbol{\phi}}_{\text{ST}}  =  {\rm\mathbf{h}}_{4,k}^{\mathcal{H}} {\boldsymbol{\Theta}}_{\text{ST}} {\rm\mathbf{H}}_1 {\rm\mathbf{x}}_{\text{SU}},  \
    \boldsymbol{\psi}_{1,k,k'} {\boldsymbol{\phi}}_{\text{PT}}  =  {\rm\mathbf{d}}_{3,k}^{\mathcal{H}} {\boldsymbol{\Theta}}_{\text{PT}} {\rm\mathbf{D}}_1 {\rm\mathbf{w}}_{\text{SD},k'}, \notag \\
    & \boldsymbol{\psi}_{2,k,k'} {\boldsymbol{\phi}}_{\text{PT}}  =  {\rm\mathbf{d}}_{3,k}^{\mathcal{H}} {\boldsymbol{\Theta}}_{\text{PT}} {\rm\mathbf{D}}_1 {\rm\mathbf{w}}_{\text{PD},k'}, \
    \boldsymbol{\psi}_{3,k} {\boldsymbol{\phi}}_{\text{PT}}  =  {\rm\mathbf{d}}_{3,k}^{\mathcal{H}} {\boldsymbol{\Theta}}_{\text{PT}} {\rm\mathbf{U}}_1 {\rm\mathbf{x}}_{\text{PU}},  \notag \\
    & \boldsymbol{\psi}_{4,k} {\boldsymbol{\phi}}_{\text{SR}}  =  {\rm\mathbf{h}}_{2,k}^{\mathcal{H}} {\boldsymbol{\Theta}}_{\text{SR}} {\rm\mathbf{H}}_1 {\rm\mathbf{x}}_{\text{SU}},   \
    \boldsymbol{\Psi}_{\text{PU}} {\boldsymbol{\phi}}_{\text{PR}}  =  {\rm\mathbf{U}}_{2} {\boldsymbol{\Theta}}_{\text{PR}} {\rm\mathbf{U}}_1 {\rm\mathbf{x}}_{\text{PU}},   \\
    & \boldsymbol{\Psi}_{\text{SU}} {\boldsymbol{\phi}}_{\text{ST}}  = {\rm\mathbf{H}}_{3} {\boldsymbol{\Theta}}_{\text{ST}} {\rm\mathbf{H}}_1 {\rm\mathbf{x}}_{\text{SU}},  \
    \boldsymbol{\Psi}_{k} {\boldsymbol{\phi}}_{\text{PR}}  =  {\rm\mathbf{U}}_{2} {\boldsymbol{\Theta}}_{\text{PR}} {\rm\mathbf{D}}_1 {\rm\mathbf{w}}_{u',k},
     \notag \\
     & d_{1,k,k'} = {\mathbf{d}}_{k}^{\mathcal{H}} {\mathbf{w}}_{\text{PD},k'}, \
     d_{2,k,k'} = {\mathbf{d}}_{k}^{\mathcal{H}} {\mathbf{w}}_{\text{SD},k'}, \\
    & v_{p,k} =  {\rm\mathbf{v}}_{p,k}^{\mathcal{H}} {\rm\mathbf{x}}_{p\text{U}}, \ \forall p\in \{ \text{P}, \text{S} \}, \
    \mathbf{u}_{\text{PU}} = \mathbf{U} \mathbf{x}_{\text{PU}}, \
    \mathbf{s}_{u',k} = {\rm\mathbf{S}} {\rm\mathbf{w}}_{u',k}. 
\end{align*}
 \endgroup
We can observe that problem $\eqref{problem_T}$ is non-convex. Accordingly, we harness the SCA for $\eqref{obj_D}$, $\eqref{C_T1}$, and $\eqref{C_T2}$ for the respective non-convex terms. With the aid of the first-order Taylor approximation, we can then arrive from problem $\eqref{problem_T}$ at a quadratic form w.r.t. $\boldsymbol{\phi}_{x}$ in $\eqref{problem_T2}$. Note that we have sorted out the related terms as second-order, first-order and constant functions for classifying the associated properties w.r.t. $\boldsymbol{\phi}_{x}$. The corresponding problem is reformulated as
 \begingroup
 \allowdisplaybreaks
 \begin{subequations} \label{problem_T2}
\begin{align}
       \mathop{\max}\limits_{\substack{{\boldsymbol{\phi}_{\text{PR}}},{\boldsymbol{\phi}_{\text{PT}}},\\ {\boldsymbol{\phi}_{\text{SR}}},{\boldsymbol{\phi}_{\text{ST}}}}} &\
       \!-\!{\boldsymbol{\phi}}_{\text{PR}}^{\mathcal{H}} \boldsymbol{\Omega}_1 {\boldsymbol{\phi}}_{\text{PR}} 
       \!+\! f_1 \left( {\boldsymbol{\phi}}_{\text{PR}} \right) 
       \!-\!{\boldsymbol{\phi}}_{\text{ST}}^{\mathcal{H}} \boldsymbol{\Omega}_2 {\boldsymbol{\phi}}_{\text{ST}} 
       \!-\!{\boldsymbol{\phi}}_{\text{PT}}^{\mathcal{H}} \boldsymbol{\Omega}_3 {\boldsymbol{\phi}}_{\text{PT}} \nonumber \\
       & \quad + f_2 \left( {\boldsymbol{\phi}}_{\text{PT}} \right) 
       \!-\!{\boldsymbol{\phi}}_{\text{SR}}^{\mathcal{H}} \boldsymbol{\Omega}_4 {\boldsymbol{\phi}}_{\text{SR}} 
       \!+\! f_3 \left( {\boldsymbol{\phi}}_{\text{SR}} \right) \label{obj_T2}\\
    \text{s.t.} \quad&  \eqref{C0}, \nonumber \\
    &
    -{\boldsymbol{\phi}}_{\text{PR}}^{\mathcal{H}} \boldsymbol{\Upsilon}_1 {\boldsymbol{\phi}}_{\text{PR}} 
    + g_1 \left( {\boldsymbol{\phi}}_{\text{PR}} \right) + c_1 \geq 0, \label{C_T2_1} \\
    & 
    -{\boldsymbol{\phi}}_{\text{PR}}^{\mathcal{H}} \boldsymbol{\Upsilon}_2 {\boldsymbol{\phi}}_{\text{PR}} 
    \!+\! g_2 \left( {\boldsymbol{\phi}}_{\text{PR}} \right) 
    \!+\! g_3 \left( {\boldsymbol{\phi}}_{\text{ST}} \right) 
    \!+\! c_2 \geq 0. \label{C_T2_2}
\end{align}
 \end{subequations}
  \endgroup
We define the related notations in problem $\eqref{problem_T2}$ at top of this page.
 \begin{figure*}
\begin{align*}
    \boldsymbol{\Omega}_1 &= \sum_{k\in \mathcal{K}_{\text{PD}}} (1+\gamma_{\text{PD},k}) \lambda_{\text{PD}, k}
\left( \sum_{k'\in \mathcal{K}_{\text{PD}}}  \boldsymbol{\varphi}_{1,k,k'}^{\mathcal{H}} \boldsymbol{\varphi}_{1,k,k'}  
	+ \sum_{k'\in \mathcal{K}_{\text{SD}}}  \boldsymbol{\varphi}_{2,k,k'}^{\mathcal{H}} \boldsymbol{\varphi}_{2,k,k'} 
	+ \boldsymbol{\varphi}_{3,k}^{\mathcal{H}} \boldsymbol{\varphi}_{3,k} \right),  \\
    \boldsymbol{\Omega}_3 & = \sum_{k\in \mathcal{K}_{\text{SD}}} (1+\gamma_{\text{SD},k}) \lambda_{\text{SD}, k}
\left( \sum_{k'\in \mathcal{K}_{\text{SD}}}  \boldsymbol{\psi}_{1,k,k'}^{\mathcal{H}} \boldsymbol{\varphi}_{1,k,k'}  
	+ \sum_{k'\in \mathcal{K}_{\text{PD}}}  \boldsymbol{\psi}_{2,k,k'}^{\mathcal{H}} \boldsymbol{\psi}_{2,k,k'} 
	+ \boldsymbol{\psi}_{3,k}^{\mathcal{H}} \boldsymbol{\psi}_{3,k} \right),  \\
    \boldsymbol{\Omega}_2 & = \sum_{k\in \mathcal{K}_{\text{PD}}} (1+\gamma_{\text{PD},k}) \lambda_{\text{PD}, k} {\boldsymbol{\varphi}}_{4,k}^{\mathcal{H}} {\boldsymbol{\varphi}}_{4,k},  \qquad
    \boldsymbol{\Omega}_4 = \sum_{k\in \mathcal{K}_{\text{SD}}} (1+\gamma_{\text{SD},k}) \lambda_{\text{SD}, k} {\boldsymbol{\psi}}_{4,k}^{\mathcal{H}} {\boldsymbol{\psi}}_{4,k},  \\
    \boldsymbol{\Upsilon}_1 &= t_{\text{PU}} \sum_{ \substack{u'\in\{\text{PD},\text{SD}\} , k\in \mathcal{K}_{u'} } } {\boldsymbol{\Psi}}_{k}^{\mathcal{H}} {\boldsymbol{\Psi}}_{k} , \qquad
    \boldsymbol{\Upsilon}_2 = t_{\text{SU}} \sum_{ \substack{u'\in\{\text{PD},\text{SD}\} , k\in \mathcal{K}_{u'} } } {\boldsymbol{\Psi}}_{k}^{\mathcal{H}} {\boldsymbol{\Psi}}_{k} , \nonumber \\
    f_1 \left( {\boldsymbol{\phi}}_{\text{PR}} \right) &= 
    \sum_{k\in \mathcal{K}_{\text{PD}}} \left( 1 + \gamma_{\text{PD},k} \right) \left\lbrace  2 \mathfrak{R} \{ \boldsymbol{\phi}_{\text{PR}}^{(a) \mathcal{H}} \boldsymbol{\varphi}_{1,k,k}^{\mathcal{H}} \boldsymbol{\varphi}_{1,k,k} \boldsymbol{\phi}_{\text{PR}} \} 
    + 2 \mathfrak{R} \{  d_{1,k,k}^{\mathcal{H}} \boldsymbol{\varphi}_{1,k,k} \boldsymbol{\phi}_{\text{PR}} \}  \right. \\
    & \left. \qquad - \lambda_{\text{PD},k} \left[ \sum_{k'\in \mathcal{K}_{\text{PD}}} 2 \mathfrak{R} \{ d_{1,k,k'}^{\mathcal{H}} \boldsymbol{\varphi}_{1,k,k'} \boldsymbol{\phi}_{\text{PR}} \}
    + \sum_{k'\in \mathcal{K}_{\text{SD}}}  2 \mathfrak{R} \{ d_{2,k,k'}^{\mathcal{H}} \boldsymbol{\varphi}_{2,k,k'} \boldsymbol{\phi}_{\text{PR}} \}
    +   2 \mathfrak{R} \{ v_{\text{P},k}^{\mathcal{H}} \boldsymbol{\varphi}_{3,k} \boldsymbol{\phi}_{\text{PR}} \} \right]
\right\rbrace,   \\
    f_2 \left( {\boldsymbol{\phi}}_{\text{PT}} \right) &= 
    \sum_{k\in \mathcal{K}_{\text{SD}}} \left( 1 + \gamma_{\text{SD},k} \right)  2 \mathfrak{R} \{ \boldsymbol{\phi}_{\text{PT}}^{(a) \mathcal{H}} \boldsymbol{\psi}_{1,k,k}^{\mathcal{H}} \boldsymbol{\psi}_{1,k,k} \boldsymbol{\phi}_{\text{PT}} \},  \qquad
    f_3 \left( {\boldsymbol{\phi}}_{\text{SR}} \right) = 
     - \sum_{k\in \mathcal{K}_{\text{SD}}} \left( 1 + \gamma_{\text{SD},k} \right)  \lambda_{\text{SD},k} 2 \mathfrak{R} \{ v_{\text{S},k}^{\mathcal{H}} \boldsymbol{\psi}_{4,k} \boldsymbol{\phi}_{\text{SR}} \} ,  \\
	g_1 \left( \boldsymbol{\phi}_{\text{PR}} \right) &= 
	2 \mathfrak{R} \{  {\boldsymbol{\phi}}_{\text{PR}}^{(a)\mathcal{H}} {\boldsymbol{\Psi}}_{\text{PU}}^{\mathcal{H}} {\boldsymbol{\Psi}}_{\text{PU}} {\boldsymbol{\phi}}_{\text{PR}}  \}  
	+ 2 \mathfrak{R} \{ {\rm\mathbf{u}}_{\text{PU}}^{\mathcal{H}} \boldsymbol{\Psi}_{\text{PU}} \boldsymbol{\phi}_{\text{PR}} \}  
	- t_{\text{PU}} \sum_{ \substack{u'\in\{\text{PD},\text{SD}\} , k\in \mathcal{K}_{u'} } } 2 \mathfrak{R} \{ {\rm\mathbf{s}}_{u',k}^{\mathcal{H}} {\boldsymbol{\Psi}}_{k} {\boldsymbol{\phi}}_{\text{PR}} \},    \\
 	g_2 \left( {\boldsymbol{\phi}}_{\text{PR}} \right) &= - t_{\text{SU}} \sum_{ \substack{u'\in\{\text{PD},\text{SD}\} , k\in \mathcal{K}_{u'} } } 2 \mathfrak{R} \{ {\rm\mathbf{s}}_{u',k}^{\mathcal{H}} {\boldsymbol{\Psi}}_{k} {\boldsymbol{\phi}}_{\text{PR}} \}, \qquad
    g_3 \left( {\boldsymbol{\phi}}_{\text{ST}} \right) =  2 \mathfrak{R} \{ {\boldsymbol{\phi}}_{\text{ST}}^{(a)\mathcal{H}} {\boldsymbol{\Psi}}_{\text{SU}}^{\mathcal{H}} {\boldsymbol{\Psi}}_{\text{SU}} {\boldsymbol{\phi}}_{\text{ST}} \}, \\    
     c_1 &= - {\boldsymbol{\phi}}_{\text{PR}}^{(a)\mathcal{H}} {\boldsymbol{\Psi}}_{\text{PU}}^{\mathcal{H}} {\boldsymbol{\Psi}}_{\text{PU}} {\boldsymbol{\phi}}_{\text{PR}}^{(a)}
     + {\rm\mathbf{u}}_{\text{PU}}^{\mathcal{H}} \boldsymbol{\Psi}_{\text{PU}} \boldsymbol{\phi}_{\text{PR}}^{(a)}
     - t_{\text{SU}} \Big( \sum_{ \substack{u'\in\{\text{PD},\text{SD}\} , k\in \mathcal{K}_{u'} } } \mathbf{s}_{u',k}^{\mathcal{H}} \mathbf{s}_{u',k} +  \sigma^2 \Big), \\
    c_2 &= - {\boldsymbol{\phi}}_{\text{ST}}^{(a)\mathcal{H}} {\boldsymbol{\Psi}}_{\text{SU}}^{\mathcal{H}} {\boldsymbol{\Psi}}_{\text{SU}} {\boldsymbol{\phi}}_{\text{ST}}^{(a)}
     - t_{\text{SU}} \Big( \sum_{ \substack{u'\in\{\text{PD},\text{SD}\} , k\in \mathcal{K}_{u'} } } \mathbf{s}_{u',k}^{\mathcal{H}} \mathbf{s}_{u',k} +  \sigma^2 \Big).
\end{align*}
\hrulefill
\end{figure*}
Note that we neglect the constant term in the objective function of $\eqref{obj_T2}$, since it does not affect the optimization. In this context, we can observe that the objective as well as the constraints of $\eqref{C_T2_1}$ and $\eqref{C_T2_2}$ are convex. However, the complete problem associated with the coupled terms of $\{\boldsymbol{\beta}, \boldsymbol{\theta}\}$ and with the different constraints in $\eqref{C0}$ is still a non-convex problem. Therefore, we partition $\eqref{problem_T2}$ into further sub-problems w.r.t. the amplitudes and phase shifts of D-STAR by defining
	\begingroup
	\allowdisplaybreaks
\begin{align}
    &{\boldsymbol{\phi}}_{x} = 
    \left[
    \beta_{x,1} e^{j\theta_{x,1}},...,
    \beta_{x,2} e^{j\theta_{x,M}}
    \right] ^{\rm T}
    \notag\\
    &=
    {\rm diag}\left( e^{j\theta_{x,1}},...,e^{j\theta_{x,M}} \right)
    \left[ \beta_{x,1},...,\beta_{x,M}\right] ^{\rm T}
    \triangleq \widetilde{\boldsymbol{\theta}}_{x} {\boldsymbol{\beta}}_{x}  \\
    &=
    {\rm diag}\left( \beta_{x,1},...,\beta_{x,M} \right)
    \left[ e^{j\theta_{x,1}},...,e^{j\theta_{x,M}} \right] ^{\rm T}
    \triangleq \widetilde{\boldsymbol{\beta}}_{x} {\boldsymbol{\theta}}_{x},
\end{align}
\endgroup
where $\widetilde{\boldsymbol{\theta}}_{x}$ and $\widetilde{\boldsymbol{\beta}}_{x}$ stand for the fixed phase shifts and amplitudes, respectively, obtained from their sub-problems.

\subsubsection{Amplitude of D-STAR}
We can reformulate $\eqref{problem_T2}$ for the amplitude part of D-STAR as
 \begingroup
 \allowdisplaybreaks
 \begin{subequations} \label{problem_b1}
\begin{align}
     \mathop{\max}\limits_{\substack{{\boldsymbol{\beta}_{\text{PR}}},{\boldsymbol{\beta}_{\text{PT}}},\\ {\boldsymbol{\beta}_{\text{SR}}},{\boldsymbol{\beta}_{\text{ST}}}}} &\quad 
     -{\boldsymbol{\beta}}_{\text{PR}}^{\mathcal{H}} \boldsymbol{\Omega}_{1,\beta} {\boldsymbol{\beta}}_{\text{PR}} 
       + f_{1,\beta} \left( {\boldsymbol{\beta}}_{\text{PR}} \right) 
       -{\boldsymbol{\beta}}_{\text{ST}}^{\mathcal{H}} \boldsymbol{\Omega}_{2,\beta} {\boldsymbol{\beta}}_{\text{ST}} \notag\\
       &\quad -{\boldsymbol{\beta}}_{\text{PT}}^{\mathcal{H}} \boldsymbol{\Omega}_{3,\beta} {\boldsymbol{\beta}}_{\text{PT}} 
        + f_{2,\beta} \left( {\boldsymbol{\beta}}_{\text{PT}} \right) 
       -{\boldsymbol{\beta}}_{\text{SR}}^{\mathcal{H}} \boldsymbol{\Omega}_{4,\beta} {\boldsymbol{\beta}}_{\text{SR}} \notag\\
       &\quad + f_{3,\beta} \left( {\boldsymbol{\beta}}_{\text{SR}} \right) \label{obj_b2} \\
    \text{s.t.} \qquad & \eqref{Con_B}, \nonumber\\
    &  
    -{\boldsymbol{\beta}}_{\text{PR}}^{\mathcal{H}} \boldsymbol{\Upsilon}_{1,\beta} {\boldsymbol{\beta}}_{\text{PR}} 
    + g_{1,\beta} \left( {\boldsymbol{\beta}}_{\text{PR}} \right) + c_1 \geq 0, \label{C_b1_1} \\
    & 
    -{\boldsymbol{\beta}}_{\text{PR}}^{\mathcal{H}} \boldsymbol{\Upsilon}_{2,\beta} {\boldsymbol{\beta}}_{\text{PR}} 
    + g_{2,\beta} \left( {\boldsymbol{\beta}}_{\text{PR}} \right) 
    + g_{3,\beta} \left( {\boldsymbol{\beta}}_{\text{ST}} \right) 
    + c_2 \geq 0, \label{C_b1_2}
\end{align}
 \end{subequations}
  \endgroup
where the notations following $\eqref{problem_T2}$ are defined as $\boldsymbol{\Omega}_{i,\beta} = \boldsymbol{\theta}^{(a) \mathcal{H}}_x \boldsymbol{\Omega}_{i} \boldsymbol{\theta}^{(a)}_x, \forall (i,x)= \{ (1, \text{PR}), (2, \text{ST}), (3, \text{PT}), (4, \text{SR})\}$ and $\boldsymbol{\Upsilon}_{i,\beta} = \boldsymbol{\theta}^{(a) \mathcal{H}}_{\text{PR}} \boldsymbol{\Upsilon}_{i} \boldsymbol{\theta}^{(a)}_{\text{PR}}, \forall i\in\{1,2\}$.
We can readily attain $f_{i,\beta}(\boldsymbol{\beta}_x), \forall (i,x)=\{(1,\text{PR}), (2,\text{PT}), (3,\text{SR})\}$ and $g_{i,\beta}(\boldsymbol{\beta}_x), \forall (i,x)=\{(1,\text{PR}), (2,\text{PT}), (3,\text{ST})\}$ upon replacing $\boldsymbol{\phi}_{x}$ with $\boldsymbol{\theta}^{(a)}_{x} \boldsymbol{\beta}_{x}$, which are neglected here, since they have similar definitions. Note that $c_1$ and $c_2$ remain unchanged constant values. Due to the quadratic equality constraint of $\eqref{Con_B}$, we alternatively apply four equivalent constraints expressed as
\begingroup
\allowdisplaybreaks
\begin{subequations}
\begin{align}
&\beta_{\text{PT},m}^2+\beta_{\text{PR},m}^2 \geq 1, \label{qq1}\\ &\beta_{\text{PT},m}^2+\beta_{\text{PR},m}^2 \leq 1, \label{qq2}\\ &\beta_{\text{ST},m}^2+\beta_{\text{SR},m}^2 \geq 1, \label{qq3}\\ &\beta_{\text{ST},m}^2+\beta_{\text{SR},m}^2 \leq 1. \label{qq4}
\end{align} 
\end{subequations}
\endgroup
As for the non-convex constraints in $\eqref{qq1}$ and $\eqref{qq3}$, we obtain their first-order Taylor approximation as
\begingroup
\allowdisplaybreaks
\begin{subequations}
\begin{align}
	&{\beta}_{\text{PT},m}^{(a)} \beta_{\text{PT},m}+ {\beta}_{\text{PR},m}^{(a)} \beta_{\text{PR},m} \geq 1, \label{beta_sca1} \\
	&{\beta}_{\text{ST},m}^{(a)} \beta_{\text{ST},m}+ {\beta}_{\text{SR},m}^{(a)} \beta_{\text{SR},m} \geq 1, \label{beta_sca2}
\end{align}
\end{subequations}
\endgroup
where $ {\beta}_{\text{PT},m}^{(a)}, {\beta}_{\text{PR},m}^{(a)}, {\beta}_{\text{ST},m}^{(a)}, {\beta}_{\text{SR},m}^{(a)}$ are solutions obtained at the $a$-th iteration. Accordingly, problem $\eqref{problem_b1}$ now becomes
 \begingroup
 \allowdisplaybreaks
 \begin{subequations} \label{problem_b11}
\begin{align}
     \mathop{\max}\limits_{\substack{{\boldsymbol{\beta}_{\text{PR}}},{\boldsymbol{\beta}_{\text{PT}}},\\ {\boldsymbol{\beta}_{\text{SR}}},{\boldsymbol{\beta}_{\text{ST}}}}} &\quad 
     -{\boldsymbol{\beta}}_{\text{PR}}^{\mathcal{H}} \boldsymbol{\Omega}_{1,\beta} {\boldsymbol{\beta}}_{\text{PR}} 
       + f_{1,\beta} \left( {\boldsymbol{\beta}}_{\text{PR}} \right) 
       -{\boldsymbol{\beta}}_{\text{ST}}^{\mathcal{H}} \boldsymbol{\Omega}_{2,\beta} {\boldsymbol{\beta}}_{\text{ST}}  \notag\\
       &\quad -{\boldsymbol{\beta}}_{\text{PT}}^{\mathcal{H}} \boldsymbol{\Omega}_{3,\beta} {\boldsymbol{\beta}}_{\text{PT}} 
       + f_{2,\beta} \left( {\boldsymbol{\beta}}_{\text{PT}} \right) 
       -{\boldsymbol{\beta}}_{\text{SR}}^{\mathcal{H}} \boldsymbol{\Omega}_{4,\beta} {\boldsymbol{\beta}}_{\text{SR}} \notag\\
       &\quad
       + f_{3,\beta} \left( {\boldsymbol{\beta}}_{\text{SR}} \right) \label{obj_b} \\
    \text{s.t.} \qquad & \eqref{C_b1_1}, \eqref{C_b1_2}, \eqref{qq2}, \eqref{qq4}, \eqref{beta_sca1}, \eqref{beta_sca2}.
\end{align}
 \end{subequations}
  \endgroup
We can infer that the problem $\eqref{problem_b11}$ is convex and can be solved by Lagrangian methods. Although the conventional Lagrangian method is a widely adopted powerful method of finding the optimum, ADMM \cite{admm, admm2} is capable of offering several advantages that make it an attractive alternative. Briefly, ADMM accomplishes better convergence, better scalability as well as a higher grade of flexibility for complex problems exhibiting non-smooth objectives and non-convex constraints. We describe the associated ADMM optimization in Appendix \ref{Appendix_ADMM}.
We define the feasible convex domains for $\boldsymbol{\beta}_{\text{PR}}$ and $\boldsymbol{\beta}_{\text{SR}}$, which are given by $\mathcal{D}_1 = \{ \eqref{C_b1_1}, \eqref{C_b1_2}, \eqref{qq2}\}$ and $\mathcal{D}_2 = \{\eqref{qq4}\}$, respectively. Moreover, we introduce the auxiliary variables $\{\varsigma_{\text{P}}, \varsigma_{\text{S}}\}$ for $\eqref{beta_sca1}$ and $\eqref{beta_sca2}$ to become equality constraints, i.e., ${\beta}_{\text{PT},m}^{(a)} \beta_{\text{PT},m}+ {\beta}_{\text{PR},m}^{(a)} \beta_{\text{PR},m} - \varsigma_{\text{P}} = 1$ and ${\beta}_{\text{ST},m}^{(a)} \beta_{\text{ST},m}+ {\beta}_{\text{SR},m}^{(a)} \beta_{\text{SR},m} - \varsigma_{\text{S}} = 1$. Therefore, according to the ADMM scheme, our alternating optimization and auxiliary parameters are updated as $\eqref{ADMMequ}$ at top of next page.
 \begin{figure*}
 \begin{subequations} \label{ADMMequ}
\begin{align}
    {\boldsymbol{\beta}}_{\text{PR}}^{(a+1)} =
    \argmax_{\boldsymbol{\beta}_{\text{PR}} \in  \mathcal{D}_1} & \quad 
    -{\boldsymbol{\beta}}_{\text{PR}}^{\mathcal{H}} {\boldsymbol{\Omega}}_{1,\beta} {\boldsymbol{\beta}}_{\text{PR}} 
    +  f_{1,\beta} \left( {\boldsymbol{\beta}}_{\text{PR}} \right) 
    - \sum_{m \in \mathcal{M}_{\text{PR}}} u_m^{(a)} \left[ \beta_{\text{PR},m}^{(a)} \beta_{\text{PR},m} -  \left( 1 +\varsigma_{\text{P}}^{(a)} - (z_{\text{PT},m}^{(a)} )^2 \right)  \right] \nonumber \\
    & \quad 
    - \rho_{1} \sum_{m\in \mathcal{M}_{\text{PR}}} \left\lVert \beta_{\text{PR},m}^{(a)} \beta_{\text{PR},m} -  \left( 1 +\varsigma_{\text{P}}^{(a)} - (z_{\text{PT},m}^{(a)} )^2 \right) \right\rVert^2;\\
    {\rm\mathbf{z}}_{\text{PT}}^{(a+1)} =
    \argmax_{{\rm\mathbf{z}_{\text{PT}}}} & \quad
    -{\rm\mathbf{z}}_{\text{PT}}^{\mathcal{H}} {\boldsymbol{\Omega}}_{3,\beta} {\rm\mathbf{z}}_{\text{PT}} 
    + f_{2,\beta} \left( {\boldsymbol{\beta}}_{\text{PT}} = {\rm\mathbf{z}}_{\text{PT}} \right)  
    -\sum_{m \in \mathcal{M}_{\text{PT}}} u_m^{(a)} \left[ z_{\text{PT},m}^{(a)} z_{\text{PT},m} -  \left( 1 + \varsigma_{\text{P}}^{(a)} - \beta_{\text{PR},m}^{(a)} \beta_{\text{PR},m}^{(a+1)}  \right)  \right] \nonumber \\
    & \quad - \rho_1 \sum_{m\in \mathcal{M}_{\text{PT}}} \left\lVert z_{\text{PT},m}^{(a)} z_{\text{PT},m} -  \left( 1 + \varsigma_{\text{P}}^{(a)} - \beta_{\text{PR},m}^{(a)} \beta_{\text{PR},m}^{(a+1)}  \right)  \right\rVert^2;\\
	\varsigma_{\text{P}}^{(a+1)}  = \argmax_{\varsigma_{\text{P}}} &  \quad -\sum_{m \in \mathcal{M}_{\text{PT}}} u_m^{(a)} \left[ z_{\text{PT},m}^{(a)} z_{\text{PT},m}^{(a+1)} -  \left( 1 + \varsigma_{\text{P}} - \beta_{\text{PR},m}^{(a)} \beta_{\text{PR},m}^{(a+1)}  \right)  \right] \nonumber \\
    & \quad - \rho_1 \sum_{m\in \mathcal{M}_{\text{PT}}} \left\lVert z_{\text{PT},m}^{(a)} z_{\text{PT},m}^{(a+1)} -  \left( 1 + \varsigma_{\text{P}} - \beta_{\text{PR},m}^{(a)} \beta_{\text{PR},m}^{(a+1)} \right)  \right\rVert^2;\\
    u_{m}^{(a+1)} = u_{m}^{(a)} - \rho_1 & \left[ z_{\text{PT},m}^{(a)} z_{\text{PT},m}^{(a+1)} - \left( 1 + \varsigma_{\text{P}}^{(a+1)} - \beta_{\text{PR},m}^{(a)} \beta_{\text{PR},m}^{(a+1)} \right) \right];\\
    {\boldsymbol{\beta}}_{\text{SR}}^{(a+1)} =
    \quad \argmax_{{\boldsymbol{\beta}}_{\text{SR}}\in \mathcal{D}_2} &\quad 
    -{\boldsymbol{\beta}}_{\text{SR}}^{\mathcal{H}}
 {\boldsymbol{\Omega}}_{4,\beta}
 {\boldsymbol{\beta}}_{\text{SR}} 
 	+ f_{3,\beta} (\boldsymbol{\beta}_{\text{SR}})
 	- \sum_{m \in \mathcal{M}_{\text{SR}}} r_m^{(a)} \left[ \beta_{\text{SR},m}^{(a)} \beta_{\text{SR},m} -  \left( 1 +\varsigma_{\text{S}}^{(a)} - (z_{\text{ST},m}^{(a)} )^2 \right)  \right] \nonumber \\
    & \quad 
    - \rho_{2} \sum_{m\in \mathcal{M}_{\text{SR}}} \left\lVert \beta_{\text{SR},m}^{(a)} \beta_{\text{SR},m} -  \left( 1 +\varsigma_{\text{S}}^{(a)} - (z_{\text{ST},m}^{(a)} )^2 \right) \right\rVert^2; \\
    {\rm\mathbf{z}}_{\text{ST}}^{(a+1)} =
    \quad \argmax_{{\rm\mathbf{z}}_{\text{ST}}} & \quad 
    -{\mathbf{z}}_{\text{ST}}^{\mathcal{H}} {\boldsymbol{\Omega}}_{2,\beta}
 {\rm\mathbf{z}}_{\text{ST}}  
    -\sum_{m \in \mathcal{M}_{\text{ST}}} r_m^{(a)} \left[ z_{\text{ST},m}^{(a)} z_{\text{ST},m} -  \left( 1 + \varsigma_{\text{S}}^{(a)} - \beta_{\text{SR},m}^{(a)} \beta_{\text{SR},m}^{(a+1)}  \right)  \right] \nonumber \\
    & \quad - \rho_2 \sum_{m\in \mathcal{M}_{\text{ST}}} \left\lVert z_{\text{ST},m}^{(a)} z_{\text{ST},m} -  \left( 1 + \varsigma_{\text{S}}^{(a)} - \beta_{\text{SR},m}^{(a)} \beta_{\text{SR},m}^{(a+1)}  \right)  \right\rVert^2;\\ 
	\varsigma_{\text{S}}^{(a+1)}  = \argmax_{\varsigma_{\text{S}}} &  \quad -\sum_{m \in \mathcal{M}_{\text{ST}}} r_m^{(a)} \left[ z_{\text{ST},m}^{(a)} z_{\text{ST},m}^{(a+1)} -  \left( 1 + \varsigma_{\text{S}} - \beta_{\text{SR},m}^{(a)} \beta_{\text{SR},m}^{(a+1)}  \right)  \right] \nonumber \\
    & \quad - \rho_2 \sum_{m\in \mathcal{M}_{\text{ST}}} \left\lVert z_{\text{ST},m}^{(a)} z_{\text{ST},m}^{(a+1)} -  \left( 1 + \varsigma_{\text{S}} - \beta_{\text{SR},m}^{(a)} \beta_{\text{SR},m}^{(a+1)} \right)  \right\rVert^2;\\
    r_{m}^{(a+1)} = r_{m}^{(a)} - \rho_2 & \left[ z_{\text{ST},m}^{(a)} z_{\text{ST},m}^{(a+1)} - \left( 1 + \varsigma_{\text{S}}^{(a+1)} - \beta_{\text{SR},m}^{(a)} \beta_{\text{SR},m}^{(a+1)} \right) \right].
\end{align}
 \end{subequations}
 \hrulefill
 \end{figure*}
Note that $\rho_1$ and $\rho_2$ represent the ADMM penalty for $\boldsymbol{\beta}_{\text{PR}}$ and $\boldsymbol{\beta}_{\text{SR}}$, respectively.

\subsubsection{Phase shifts of D-STAR}
After obtaining amplitudes of D-STAR, we proceed to attain its optimal phase shifts. Similar to that in problem $\eqref{problem_b1}$, we can reformulate $\eqref{problem_T2}$ for phase shift part of D-STAR as
 \begingroup
 \allowdisplaybreaks
 \begin{subequations} \label{problem_t1}
\begin{align}
     \mathop{\max}\limits_{\substack{{\boldsymbol{\theta}_{\text{PR}}},{\boldsymbol{\theta}_{\text{PT}}},\\{\boldsymbol{\theta}_{\text{SR}}},{\boldsymbol{\theta}_{\text{ST}}}}} &\quad 
     -{\boldsymbol{\theta}}_{\text{PR}}^{\mathcal{H}} \boldsymbol{\Omega}_{1,\theta} {\boldsymbol{\theta}}_{\text{PR}} 
       + f_{1,\theta} \left( {\boldsymbol{\theta}}_{\text{PR}} \right) 
       -{\boldsymbol{\theta}}_{\text{ST}}^{\mathcal{H}} \boldsymbol{\Omega}_{2,\theta} {\boldsymbol{\beta}}_{\text{ST}} \notag\\
       &\quad
       -{\boldsymbol{\theta}}_{\text{PT}}^{\mathcal{H}} \boldsymbol{\Omega}_{3,\theta} {\boldsymbol{\theta}}_{\text{PT}} 
       + f_{2,\theta} \left( {\boldsymbol{\theta}}_{\text{PT}} \right) 
       -{\boldsymbol{\theta}}_{\text{SR}}^{\mathcal{H}} \boldsymbol{\Omega}_{4,\theta} {\boldsymbol{\theta}}_{\text{SR}} \notag\\
       & \quad + f_{3,\theta} \left( {\boldsymbol{\theta}}_{\text{SR}} \right) \label{obj_t} \\
    \text{s.t.} \quad & \eqref{Con_T},\eqref{Con_T1}, \nonumber\\
    &  
    -{\boldsymbol{\theta}}_{\text{PR}}^{\mathcal{H}} \boldsymbol{\Upsilon}_{1,\theta} {\boldsymbol{\theta}}_{\text{PR}} 
    + g_{1,\theta} \left( {\boldsymbol{\theta}}_{\text{PR}} \right) + c_1 \geq 0, \label{C_t1_1} \\
    & 
    -{\boldsymbol{\theta}}_{\text{PR}}^{\mathcal{H}} \boldsymbol{\Upsilon}_{2,\theta} {\boldsymbol{\theta}}_{\text{PR}} 
    + g_{2,\theta} \left( {\boldsymbol{\theta}}_{\text{PR}} \right) 
    + g_{3,\theta} \left( {\boldsymbol{\theta}}_{\text{ST}} \right) 
    + c_2 \geq 0, \label{C_t1_2}
\end{align}
 \end{subequations}
  \endgroup
where the notations following $\eqref{problem_T2}$ are defined as $\boldsymbol{\Omega}_{i,\theta} = \boldsymbol{\beta}^{(a+1) \mathcal{H}}_x \boldsymbol{\Omega}_{i} \boldsymbol{\beta}^{(a+1)}_x, \forall (i,x)= \{ (1, \text{PR}), (2, \text{ST}), (3, \text{PT}), (4, \text{SR})\}$ and $\boldsymbol{\Upsilon}_{i,\theta} = \boldsymbol{\beta}^{(a+1) \mathcal{H}}_{\text{PR}} \boldsymbol{\Upsilon}_{i} \boldsymbol{\beta}^{(a+1)}_{\text{PR}}, \forall i\in\{1,2\}$. We can readily attain $f_{i,\theta}(\boldsymbol{\theta}_x), \forall (i,x)=\{(1,\text{PR}), (2,\text{PT}), (3,\text{SR})\}$ and $g_{i,\theta}(\boldsymbol{\theta}_x), \forall (i,x)=\{(1,\text{PR}), (2,\text{PT}), (3,\text{ST})\}$ by replacing $\boldsymbol{\phi}_{x}$ with $\boldsymbol{\beta}^{(a+1)}_{x} \boldsymbol{\theta}_{x}$, which are neglected here due to similar definitions. It is worth mentioning that the new solutions for $\boldsymbol{\beta}_{x}$ are applied based on the optimum solution $\eqref{problem_b1}$ at next iteration $(a+1)$.

We can observe that $\eqref{Con_T}$ and $\eqref{Con_T1}$ lead to an unsolvable problem. Since $\eqref{Con_T1}$ is only valid for the hardware constraints associated with coupled phase shifts, we first consider only the generic constraint of $\eqref{Con_T}$. It can be seen that $\eqref{Con_T}$ is non-convex, which should be further processed. We adopt a PCCP mechanism \cite{pccp} to obtain a convex problem. A pair of quadratic bounds as well as the non-negative penalty of $\mathbf{b}_x =\{ b_{x,m} | \forall x \in \mathcal{X}, m \in \mathcal{M}_{x} \}$ are introduced for this term, i.e., $|\theta_{x,m}|^2\geq 1-b_{x,m}$ and $|\theta_{x,m}|^2\leq 1+b_{x,m}$, forming two circular manifolds having different sizes. We can obtain the transformed problem associated with a penalty as
 \begingroup
 \allowdisplaybreaks
 \begin{subequations} \label{problem_t22}
\begin{align}
    \mathop{\max}\limits_{\substack{{\boldsymbol{\theta}}_{\text{PR}},{\boldsymbol{\theta}}_{\text{PT}},\\{\boldsymbol{\theta}}_{\text{SR}},{\boldsymbol{\theta}}_{\text{ST}}, {\rm\mathbf{b}}}} &\quad  
       -{\boldsymbol{\theta}}_{\text{PR}}^{\mathcal{H}} \boldsymbol{\Omega}_{1,\theta} {\boldsymbol{\theta}}_{\text{PR}} 
       + f_{1,\theta} \left( {\boldsymbol{\theta}}_{\text{PR}} \right) 
       -{\boldsymbol{\theta}}_{\text{ST}}^{\mathcal{H}} \boldsymbol{\Omega}_{2,\theta} {\boldsymbol{\beta}}_{\text{ST}} \notag\\
       &\quad
       -{\boldsymbol{\theta}}_{\text{PT}}^{\mathcal{H}} \boldsymbol{\Omega}_{3,\theta} {\boldsymbol{\theta}}_{\text{PT}} 
       + f_{2,\theta} \left( {\boldsymbol{\theta}}_{\text{PT}} \right) 
       -{\boldsymbol{\theta}}_{\text{SR}}^{\mathcal{H}} \boldsymbol{\Omega}_{4,\theta} {\boldsymbol{\theta}}_{\text{SR}} \notag\\
       &\quad + f_{3,\theta} \left( {\boldsymbol{\theta}}_{\text{SR}} \right)   
       - \kappa^{(a)}\sum_{\substack{ m\in \mathcal{M}_x, x\in \mathcal{X}}} b_{x,m} \label{obj_t22} \\
    \text{s.t.} &\quad  \eqref{C_t1_1}, \eqref{C_t1_2},  \nonumber\\
    & \quad \lvert \theta_{x,m} \rvert^2 \geq 1- b_{x,m},  \ \forall x \in \mathcal{X}, m \in \mathcal{M}_{x},  \label{C_t22_2} \\
    & \quad \lvert \theta_{x,m} \rvert^2 \leq 1 + b_{x,m},  \ \forall x \in \mathcal{X}, m \in \mathcal{M}_{x},  \label{C_t22_3} \\
    & \quad {\rm\mathbf{b}}_x \succeq 0,  \ \forall x\in \mathcal{X}, \label{C_t22_4}
\end{align}
 \end{subequations}
  \endgroup
where $\kappa^{(a)}$ is the PCCP penalty. Since $\eqref{C_t22_2}$ is non-convex, the classic Taylor approximation is harnessed, which yields the optimization problem
 \begingroup
 \allowdisplaybreaks
 \begin{subequations} \label{problem_t2}
\begin{align}
    \mathop{\max}\limits_{\substack{{\boldsymbol{\theta}}_{\text{PR}},{\boldsymbol{\theta}}_{\text{PT}},\\{\boldsymbol{\theta}}_{\text{SR}},{\boldsymbol{\theta}}_{\text{ST}}, {\rm\mathbf{b}}}} &\quad  
       -{\boldsymbol{\theta}}_{\text{PR}}^{\mathcal{H}} \boldsymbol{\Omega}_{1,\theta} {\boldsymbol{\theta}}_{\text{PR}} 
       + f_{1,\theta} \left( {\boldsymbol{\theta}}_{\text{PR}} \right) 
       -{\boldsymbol{\theta}}_{\text{ST}}^{\mathcal{H}} \boldsymbol{\Omega}_{2,\theta} {\boldsymbol{\beta}}_{\text{ST}} \notag\\
       &\quad
       -{\boldsymbol{\theta}}_{\text{PT}}^{\mathcal{H}} \boldsymbol{\Omega}_{3,\theta} {\boldsymbol{\theta}}_{\text{PT}} 
       + f_{2,\theta} \left( {\boldsymbol{\theta}}_{\text{PT}} \right) 
       -{\boldsymbol{\theta}}_{\text{SR}}^{\mathcal{H}} \boldsymbol{\Omega}_{4,\theta} {\boldsymbol{\theta}}_{\text{SR}} \notag\\
       &\quad
       + f_{3,\theta} \left( {\boldsymbol{\theta}}_{\text{SR}} \right) \label{obj_t2}  - \kappa^{(a)}\sum_{\substack{ m\in \mathcal{M}_x, x\in \mathcal{X}}} b_{x,m} \\
    \text{s.t.} &\quad  \eqref{C_t1_1}, \eqref{C_t1_2}, \eqref{C_t22_4},  \nonumber \\
    & \quad \mathfrak{R} \left\{ \left( \theta_{x,m}^{(a)} \right)^{\mathcal{H}} \theta_{x,m} \right\} \geq 1-b_{x,m},  \ \forall x \in \mathcal{X}, m \in \mathcal{M}_{x},  \label{C_t2_2}\\
    & \quad \lvert \theta_{x,m} \rvert^2 \leq 1 + b_{x,m},  \ \forall x \in \mathcal{X}, m \in \mathcal{M}_{x}.  \label{C_t2_3}
\end{align}
 \end{subequations}
  \endgroup
Therefore, we can obtain the optimum solution of the convex problem $\eqref{problem_t2}$. As for the coupled phase shifts, we should consider $\mathcal{M}_{\text{PT}}=\mathcal{M}_{\text{PR}}, \mathcal{M}_{\text{ST}}=\mathcal{M}_{\text{SR}}$ and $\eqref{Con_T1}$. Then, the problem becomes
\begingroup
\allowdisplaybreaks
\begin{subequations} \label{problem_t3}
\begin{align}
    \mathop{\max}\limits_{\substack{{\boldsymbol{\theta}}_{\text{PR}},{\boldsymbol{\theta}}_{\text{PT}},\\{\boldsymbol{\theta}}_{\text{SR}},{\boldsymbol{\theta}}_{\text{ST}}}} &\quad  
       -{\boldsymbol{\theta}}_{\text{PR}}^{\mathcal{H}} \boldsymbol{\Omega}_{1,\theta} {\boldsymbol{\theta}}_{\text{PR}} 
       + f_{1,\theta} \left( {\boldsymbol{\theta}}_{\text{PR}} \right) 
       -{\boldsymbol{\theta}}_{\text{ST}}^{\mathcal{H}} \boldsymbol{\Omega}_{2,\theta} {\boldsymbol{\beta}}_{\text{ST}}  \notag\\
       &\quad
       -{\boldsymbol{\theta}}_{\text{PT}}^{\mathcal{H}} \boldsymbol{\Omega}_{3,\theta} {\boldsymbol{\theta}}_{\text{PT}} 
       + f_{2,\theta} \left( {\boldsymbol{\theta}}_{\text{PT}} \right) 
       -{\boldsymbol{\theta}}_{\text{SR}}^{\mathcal{H}} \boldsymbol{\Omega}_{4,\theta} {\boldsymbol{\theta}}_{\text{SR}}  \notag\\
       &\quad
       + f_{3,\theta} \left( {\boldsymbol{\theta}}_{\text{SR}} \right) \label{obj_t3}  \\
    \text{s.t.} &\quad  \eqref{Con_T1},\nonumber\\
    &\quad \mathcal{M}_{\text{PT}}=\mathcal{M}_{\text{PR}}, \mathcal{M}_{\text{ST}}=\mathcal{M}_{\text{SR}}.
\end{align}
 \end{subequations}
 \endgroup
After obtaining the solutions of $\eqref{problem_t2}$, we heuristically compare the conditions $\eqref{Con_T1}$ in $\eqref{problem_t3}$ to the following set:
\begin{align} \label{S}
& (\theta_{t,m}^*, \theta_{r,m}^*) = \left\lbrace (\theta_{t,m}, \pm j \theta_{t,m}), ( \pm j \theta_{r,m}, \theta_{r,m}) \right\rbrace, \notag \\
&\qquad \forall (t,r)=\{(\text{PT}, \text{PR}), (\text{ST}, \text{SR})\}, x \in \mathcal{X}, m \in \mathcal{M}_{x}.
\end{align}
There are four possible cases to be compared in $\eqref{S}$. We compare the selected element of D-STAR for providing the highest objective value $\eqref{obj_t3}$, while keeping the remaining arguments fixed. The iterative comparison continues until no significant improvement of the objective value is obtained. We note that using exhaustive search may be inappropriate as it imposes an unaffordable computational complexity for a large number of D-STAR elements. The concrete DBAP algorithm of D-STAR is summarized in Algorithm \ref{alg1}. We solve the respective sub-problems for the active beamforming in $\eqref{problem_w2}$, for the amplitudes in $\eqref{problem_b1}$, and for the phase shifts in $\eqref{problem_t2}$. Additional updates will be performed in $\eqref{problem_t3}$ if phase shifts are coupled. Note that convergence is achieved when $|R_{\text{PD}}^{(a)} + R_{\text{SD}}^{(a)} - R_{\text{PD}}^{(a-1)} - R_{\text{SD}}^{(a-1)}| \leq \delta_{R}$, $\lVert\boldsymbol{\Xi}^{(a)} - \boldsymbol{\Xi}^{(a-1)} \rVert^2 \leq \delta_{\Xi}, \forall \boldsymbol{\Xi}\in\{ \mathbf{w}_{\text{PD}}, \mathbf{w}_{\text{SD}}, \boldsymbol{\beta}, \boldsymbol{\theta}\}$, or $a\geq I_{th}$, where $I_{th}$ is the maximum affordable number of iterations.

\subsection{Convergence Analysis}

The proof of convergence of Lemma \ref{lemma_DIN} can be found in Proposition 2 of \cite{converge}, while the convergence of Lemma \ref{lemma_LD} can be readily derived by following the same process as Lemma \ref{lemma_DIN}, which is omitted here. As for Lemma \ref{lemma_sca}, we consider the objective function $f(\boldsymbol{\Xi})$ partitioned into a function having concave plus convex term as $f(\boldsymbol{\Xi}) = f^+(\boldsymbol{\Xi}) + f^-(\boldsymbol{\Xi})$, where $\boldsymbol{\Xi}=\{{\mathbf{w}}_{\text{PD}},{\mathbf{w}}_{\text{SD}}, \boldsymbol{\Theta}\}$ represents the total candidate solution set. Owing to the convex nature of the function of $f^-(\boldsymbol{\Xi})$, we can use the first-order Taylor approximation of $f(\boldsymbol{\Xi})$ as the lower bound of its original objective. Then, we have
	\begingroup
	\allowdisplaybreaks
	 \begin{align}
        & f(\boldsymbol{\Xi}^{(a)}) 
        = f^+(\boldsymbol{\Xi}^{(a)}) - 
        f^-(\boldsymbol{\Xi}^{(a)}) \notag \\
        & \geq f^+(\boldsymbol{\Xi}^{(a)}) - 
        f^-(\boldsymbol{\Xi}^{(a-1)}) \notag \\
        & \qquad\qquad - \sum_{\boldsymbol X \in \boldsymbol{\Xi}} \nabla^{\mathcal{H}}_{\boldsymbol{X}} f^-(\boldsymbol{\Xi}^{(a-1)}) \cdot \left({\boldsymbol X}^{(a)}-{\boldsymbol X}^{(a-1)} \right) \notag\\
	    & = \mathop{\max}_{\boldsymbol{\Xi}\in\mathcal{D}}
        f^+(\boldsymbol{\Xi}) - 
        f^-(\boldsymbol{\Xi}^{(a-1)}) \notag \\
        & \qquad\qquad \sum_{\boldsymbol X \in \boldsymbol{\Xi}} \nabla^{\mathcal{H}}_{{\boldsymbol X}} F^-(\boldsymbol{\Xi}^{(a-1)}) \cdot \Big(\boldsymbol{X} - \boldsymbol{X}^{(a-1)} \Big) \notag\\
        & \geq f^+(\boldsymbol{\Xi}^{(a-1)}) - 
        f^-(\boldsymbol{\Xi}^{(a-1)}) \notag \\
        & \qquad\qquad \sum_{\boldsymbol X \in \boldsymbol{\Xi}} \nabla^{\mathcal{H}}_{{\boldsymbol X}} f^-(\boldsymbol{\Xi}^{(a-1)}) \cdot \Big(\boldsymbol{X}^{(a-1)} - \boldsymbol{X}^{(a-1)}\Big)\notag \\
        & 
        = f(\boldsymbol{\Xi}^{(a-1)}).
        \end{align}
        \endgroup
Note that $\mathcal{D}$ is defined as the feasible domain of variable $\boldsymbol{\Xi}$. Based on the relationship stated above, we conclude that the optimal value will be improved or stay unchanged. We also note that the transformed convex constraints guarantee the lower bounds of the original non-convex constraints. The solutions of the coupled phase shifts can be regarded as the quantized values leading to a degraded throughput. When $a\to \infty$, we can acquire the optimum by $\boldsymbol{\Xi}^*= \lim_{a \to \infty} f(\boldsymbol{\Xi}^{(a)})$. That is, the solutions obtained by the respective problems $\eqref{problem_w2}$, $\eqref{problem_b11}$, and $\eqref{problem_t2}$ w.r.t. $\{{\mathbf{w}}_{\text{PD}},{\mathbf{w}}_{\text{SD}}, \boldsymbol{\Theta}\}$ corresponding to DBAP in Algorithm \ref{alg1} will converge.

\begin{algorithm}[!tb]
\small
  \caption{Proposed DBAP scheme in D-STAR}
  \SetAlgoLined
  \DontPrintSemicolon
  \label{alg1}
  \begin{algorithmic}[1]
   \STATE Randomly initialize temporary solutions $\{\mathbf{w}_{\text{PD}}^{(a)}, \mathbf{w}_{\text{SD}}^{(a)}, \boldsymbol{\beta}_{x}^{(a)}, \boldsymbol{\theta}_{x}^{(a)} \}, \forall x\in \mathcal{X}$
   \STATE Set iteration $a=1$
	\WHILE{not converged}   
		\STATE Solve problem $\eqref{problem_w2}$ for $\{ \mathbf{w}_{\text{PD}}^{(a+1)}, \mathbf{w}_{\text{SD}}^{(a+1)} \}$ based on $\{\mathbf{w}_{\text{PD}}^{(a)}, \mathbf{w}_{\text{SD}}^{(a)}, \boldsymbol{\beta}_{x}^{(a)}, \boldsymbol{\theta}_{x}^{(a)} \}$
		\STATE Solve problem $\eqref{problem_b11}$ for $\{\boldsymbol{\beta}_{x}^{(a+1)}\}$ based on $\{\mathbf{w}_{\text{PD}}^{(a+1)}, \mathbf{w}_{\text{SD}}^{(a+1)}, \boldsymbol{\beta}_{x}^{(a)}, \boldsymbol{\theta}_{x}^{(a)} \}$
		\STATE Solve problem $\eqref{problem_t2}$ for $\{\boldsymbol{\theta}_{x}^{(a+1)}\}$ based on $\{\mathbf{w}_{\text{PD}}^{(a+1)}, \mathbf{w}_{\text{SD}}^{(a+1)}, \boldsymbol{\beta}_{x}^{(a+1)}, \boldsymbol{\theta}_{x}^{(a)} \}$
		\STATE Update coupled phase shifts by $\eqref{S}$
		\STATE Update auxiliary parameters $\{ \gamma_{u,k}^{(a+1)}, \lambda_{u,k}^{(a+1)} \}$
		\STATE Update iteration $a\leftarrow a+1$
	\ENDWHILE
	\STATE \textbf{Return} Optimum D-STAR configuration $\{\mathbf{w}_{\text{PD}}^*, \mathbf{w}_{\text{SD}}^*, \boldsymbol{\beta}_{x}^*, \boldsymbol{\theta}_{x}^*\}, \forall x\in \mathcal{X}$
  \end{algorithmic}
\end{algorithm}

\section{Numerical Results}\label{NR}

\begin{table}[!t]
\centering
\footnotesize
  \caption{Parameter Setting of D-STAR}
 \begin{tabular}{l*{1}{l}r}
  \hline
  \textbf{System Parameter} & \textbf{Value} \\
  \hline
  Distance between BS-D-STAR  & $100$ m    \\
  Distance between user-D-STAR  & $30$ m    \\
  Distance between user-BS  & $80$ m \\  
  Inter-D-STAR distance  & $100$ m    \\
  Inter-user group distance  & $100$ m   \\
  BS transmit/receiving antennas & $[8,24]$ \\
  Number of user antenna & $1$ \\
  Number of total/per-group users & $8, 2$ \\
  BS/user transmit power & $30, 20$ dBm \\
  Maximum power constraint & $40$ dBm \\
  Number of STAR-RISs in D-STAR & $2$\\
  Number of elements per STAR-RIS & $[8, 24]$\\
  UL rate requirement  & $0.5$ bps/Hz  \\
  Noise power & $-80$ dBm      \\
  ADMM penalty term & $1$        \\
  PCCP penalty term & $0.1$        \\
  Convergence thresholds & $10^{-3}$ \\
  Iteration upper bounds & $20$ \\
  Monte Carlo runs & $100$ \\
  \hline
 \end{tabular} \label{Parameter}
\end{table}


\begin{figure}[!t]
	\centering
	\includegraphics[width=3in]{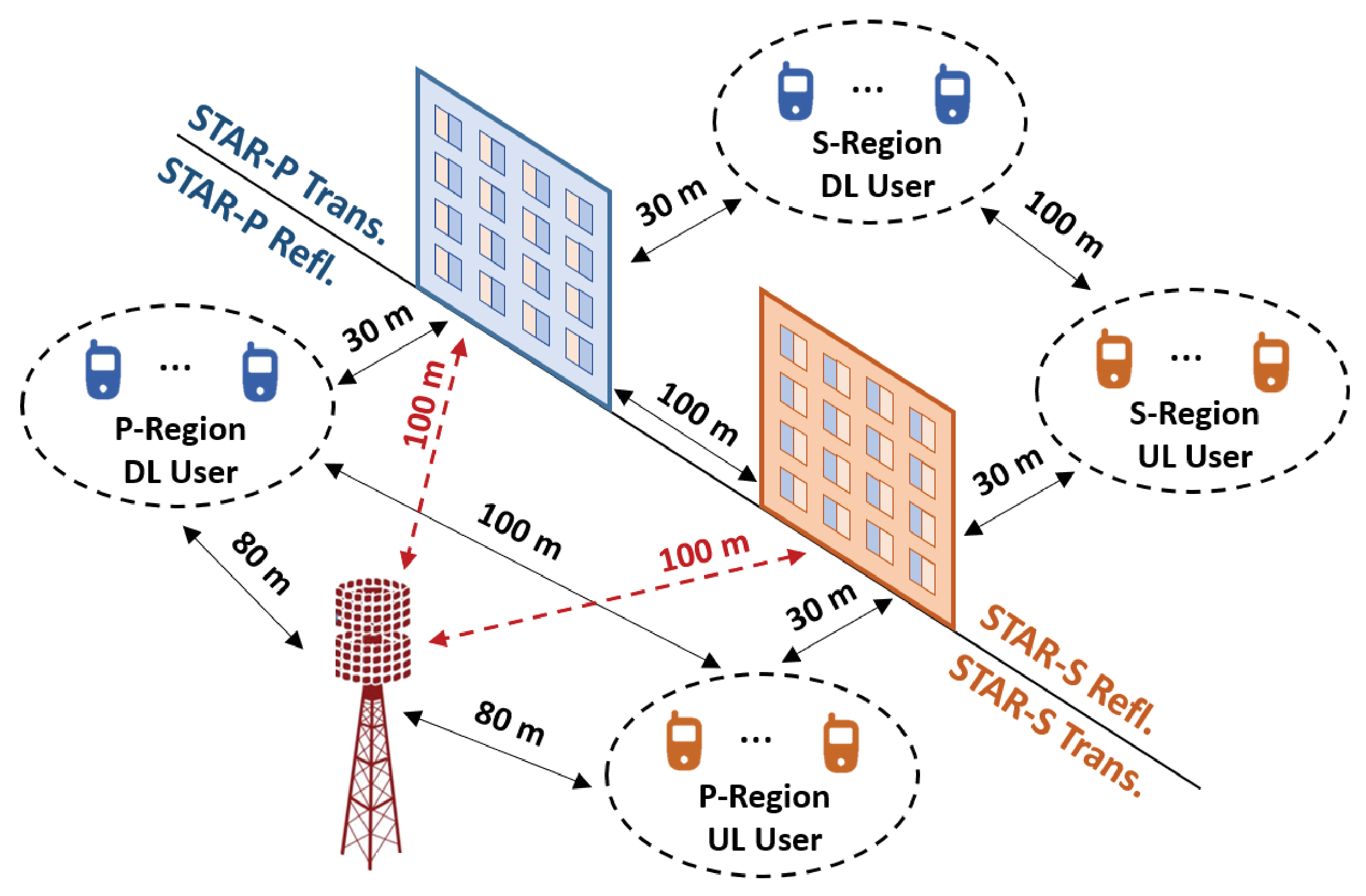}
 \caption{The relative distances of the deployed D-STAR architecture.} \label{dist}
\end{figure}

\begin{figure}[!t]
	\centering
	\includegraphics[width=3in]{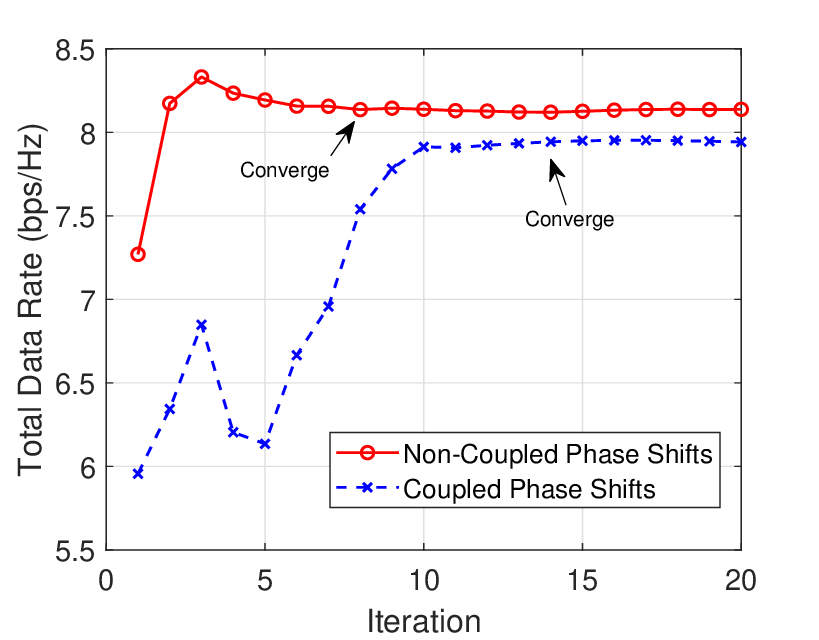}
 \caption{Convergence of the proposed DBAP scheme both with and without coupled phase shifts.} \label{conv}
\end{figure}

%

\begin{figure}[!t]
	\centering
	\includegraphics[width=3in]{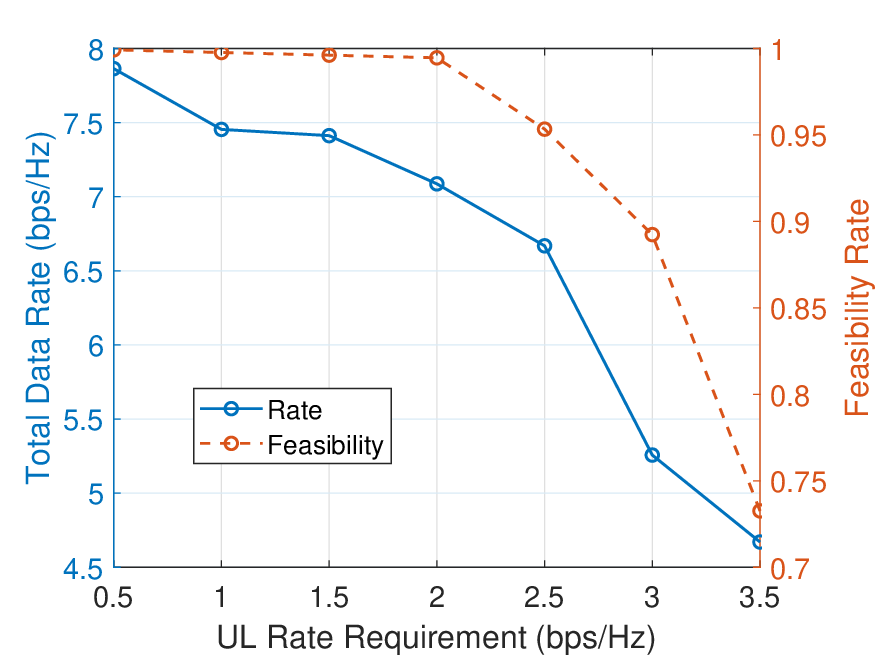}
 \caption{Rate and feasiblity performance versus different UL rate requirements.} \label{qos}
\end{figure}

We consider a single JUD BS serving four groups of users with each having $K_{u}=2$ users $\forall u\in \mathcal{U}$. The relative distances of the BS, D-STAR and the users are depicted in Fig. \ref{dist}. The BS is equipped with $[8,24]$ transmit and receive antennas, whilst the user equipment is equipped with a single antenna. The transmit power of the BS and of the user is set to $30$ and $20$ dBm, respectively. The channel follows Rician fading \cite{star_cite2}, including the deterministic line-of-sight (LoS) components of the array response and the NLoS components modeled as Rayleigh fading. Note that the direct link possesses much more LoS components than NLoS paths, whilst only NLoS components are considered between the BS and the D-STAR as well as between the D-STAR and the users. We use an identical number of elements for the two STAR-RISs, i.e., $M_x=M, \forall x\in \mathcal{X}$. The UL rate requirement is set to $1$ bps/Hz. The channels obeys uncorrelated Rayleigh fading with noise power of $-80$ dBm. Moreover, we utilize the popular optimization tool CVX \cite{cvx} as our optimal policy for our DBAP scheme in the context of the D-STAR. All simulations are averaged over $100$ Monte Carlo runs. The remaining system parameters of the proposed D-STAR architecture are listed in Table \ref{Parameter}. In Fig. \ref{conv}, we can observe that the proposed DBAP scheme in D-STAR converges both with and without coupled phase shifts. We can infer from the figure that slower convergence is attained in a coupled-phase scenario owing to the quantized nature of the solutions in $\eqref{S}$, imposing a modest rate degradation of less than $2\%$. Moreover, in Fig. \ref{qos}, we evaluate the feasibility for different UL rate constraints, where the feasibility rate is defined as the probability of the results satisfying UL constraints. We can observe that more stringent service requirement will lead to a reduced data rate as well as fewer feasible solutions owing to the increasingly insufficient resources in terms of D-STAR elements and BS antennas.

\begin{figure}[!t]
	\centering
	\includegraphics[width=3in]{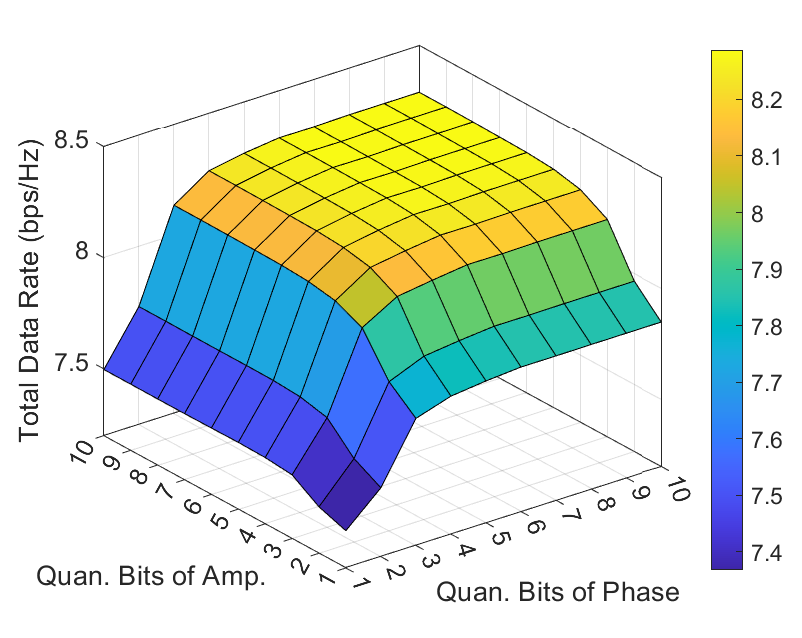}
	\caption{The performance of DBAP in D-STAR for joint quantization in both amplitude/phase w.r.t. different quantization $N_A, N_P \in [1,10]$ bits.} \label{fig:bit_3d}
\end{figure}

\begin{figure}[!t]
	\centering
	\includegraphics[width=3 in]{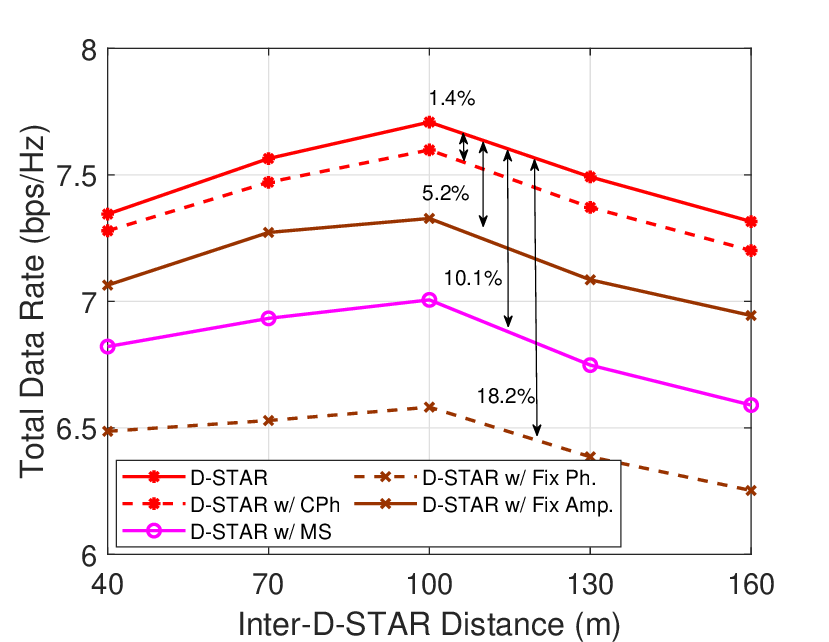}
		\caption{The performance of D-STAR w.r.t. different inter-D-STAR distance $\{ 40, 70, 100, 130, 160\}$ m. We compare D-STAR to its sub-schemes with CPh, MS, and optimization under fixed phase/amplitudes.}
	\label{fig:ios_dis}
\end{figure}

\subsection{Quantization Effects}
As demonstrated in Fig. \ref{fig:bit_3d}, we evaluate the proposed DBAP in the context of the D-STAR architecture in terms of different number of quantization bits. Considering $N_A$ bits for the amplitude and $N_P$ bits for the phase shifts, we have the solution range of $\frac{1}{2^{N_A}}\cdot [0,1,...,2^{N_A}-1]$ and $\frac{2\pi}{2^{N_P}}\cdot [0,1,...,2^{N_P}-1]$, respectively. Accordingly, the quantized solution can be expressed by
\begin{align}
	\beta_{x,m}^{(Q)} &= \left\lfloor \beta_{x,m}\cdot 2^{N_A} \right\rfloor \cdot \frac{1}{2^{N_A}},  \\
	\theta_{x,m}^{(Q)} &= \left\lfloor \frac{\theta_{x,m}}{2\pi} \cdot 2^{N_P} \right\rfloor \cdot \frac{2 \pi}{2^{N_P}}.
\end{align}
We can observe in Fig. \ref{fig:bit_3d} that even a relatively low amplitude and phase-resolution only imposes limited data-rate reduction, while having a low hardware complexity. We can also infer from the result that the quantization effect is more detrimental to the phase shifts than the amplitude, i.e., low-resolution phases will significantly degrade the rate owing to its sensitivity.

\subsection{Deployment of D-STAR}

As demonstrated in Fig. \ref{fig:ios_dis}, we have studied the critical issue of D-STAR deployment in terms of adjusting the inter-D-STAR distances between $40$ and $160$ m. We compare the proposed D-STAR architecture to its relatives associated with coupled phase shifts (CPh), the MS mechanism, and to optimization under fixed phase (Ph.) as well as fixed amplitudes (Amp.). Intriguingly, we can observe that the curves exhibit a concave shape in conjunction with the optimum inter-D-STAR distance of $100$ m. When the distance is in a range of $[40, 100]$ m, increased throughput is observed, since the STAR-RIS scheme improves the beamforming directivity. As a benefit of concentrating the power in the desirable directions, the interferences impinging from the SU, from the STAR-S to SD link as well as from the PU, and from the STAR-P to PD link are substantially alleviated. The optimum distance happens to be $100$ m, since this corresponds to the minimum distance from per D-STAR w.r.t. the corresponding user groups. Upon further extending the distances to $[100, 160]$ m, the weaker signals will eventually result in a reduced rate.

\subsection{Partitioning D-STAR}

\begin{figure*}[!t]
		\centering
		\subfigure[]
		{\includegraphics[width=3 in]{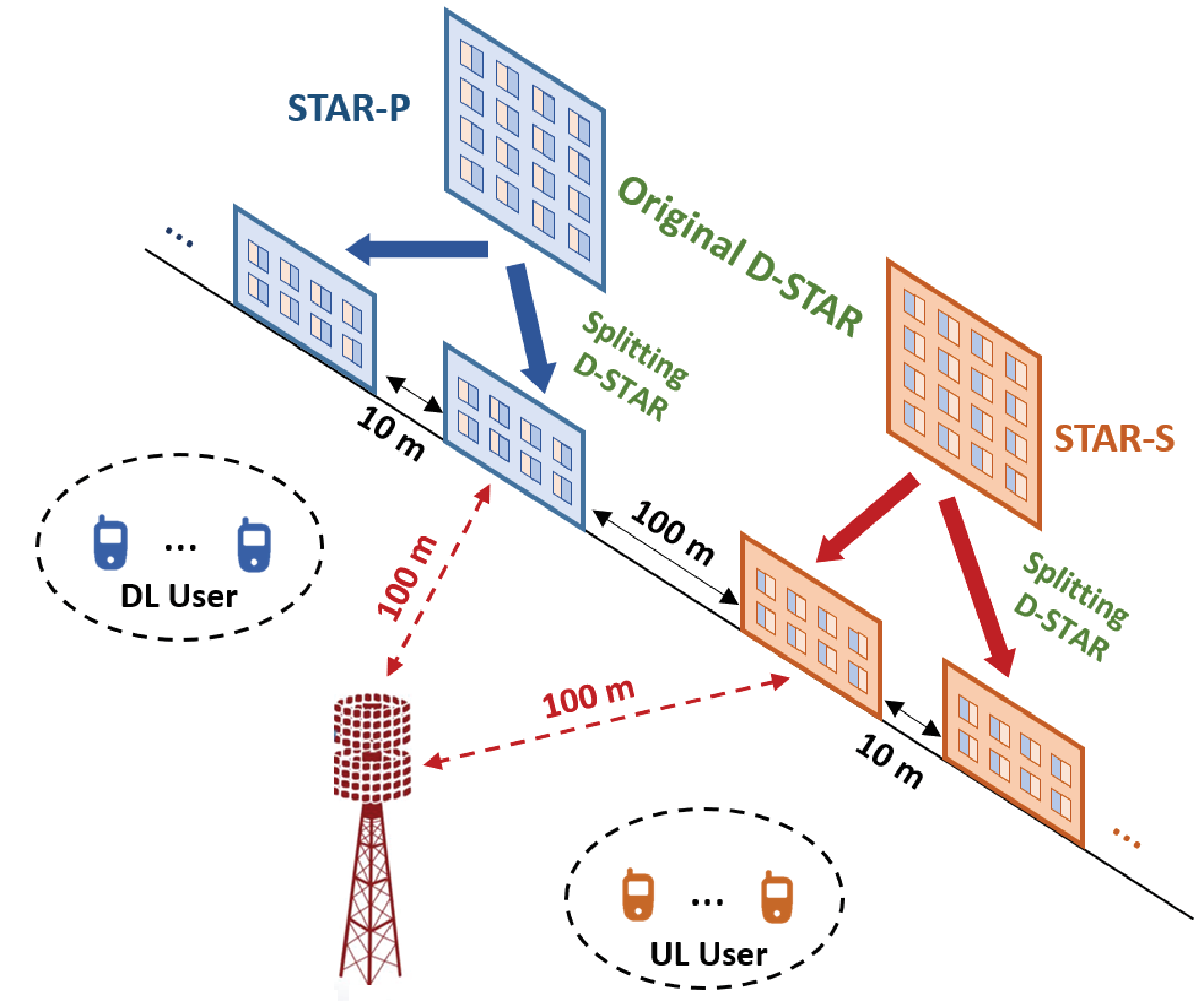} \label{split}}
		\subfigure[]
		{\includegraphics[width=3 in]{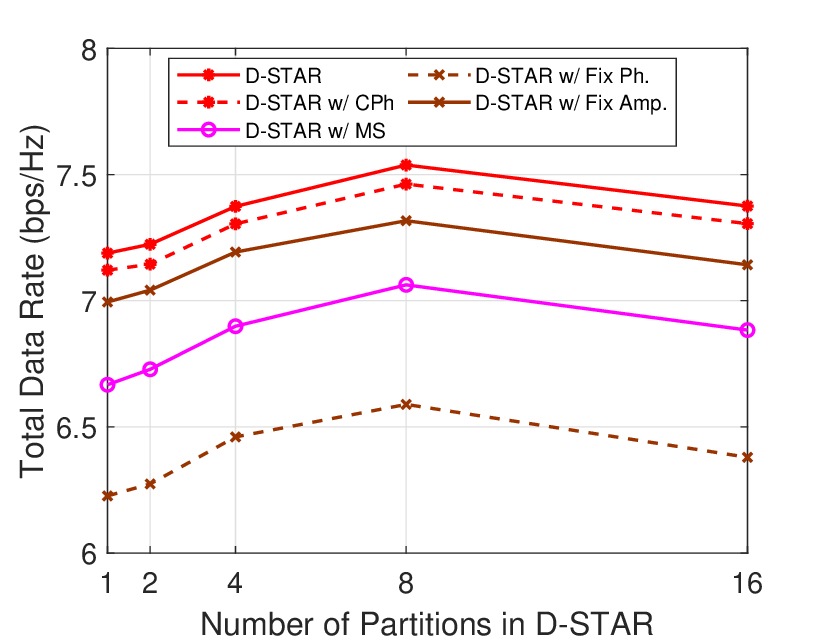} \label{fig:ios_split}}
		\caption{(a) Example for the architecture of splitting D-STAR into several sub-D-STARs. (b) Performance of D-STAR w.r.t. different numbers of splitting D-STARs. We compare D-STAR to its sub-schemes with CPh, MS, and optimization under fixed phase/amplitudes.}
		\label{fig:split}
	\end{figure*}

\begin{figure*}[!t]
		\centering
		\subfigure[]
		{\includegraphics[width=2.1 in]{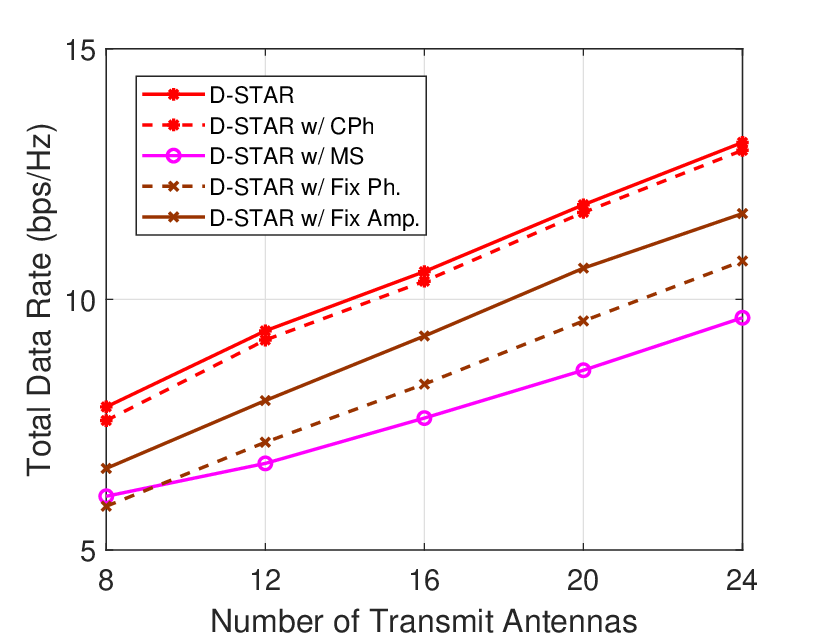} \label{fig:dios_ant}}
		\subfigure[]
		{\includegraphics[width=2.1 in]{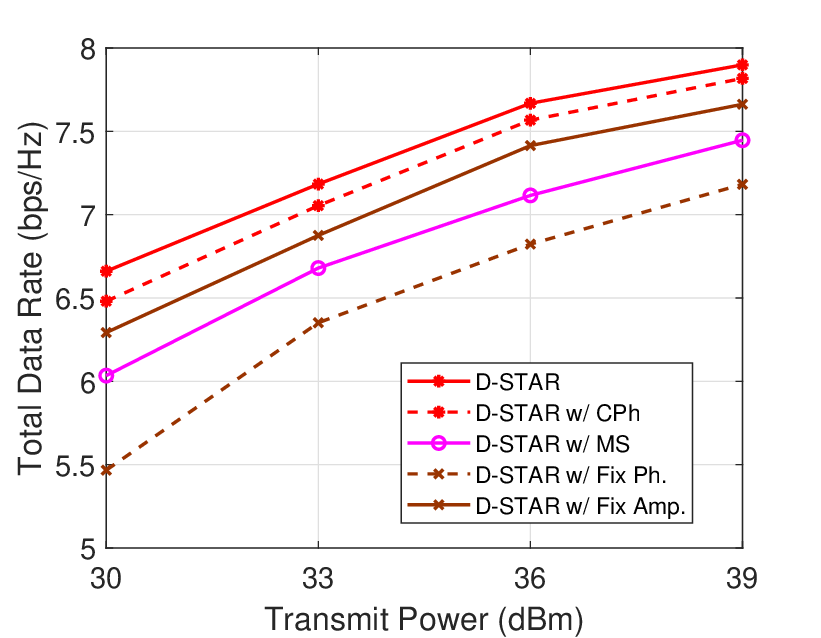} \label{fig:dios_power}}
		\subfigure[]
		{\includegraphics[width=2.1 in]{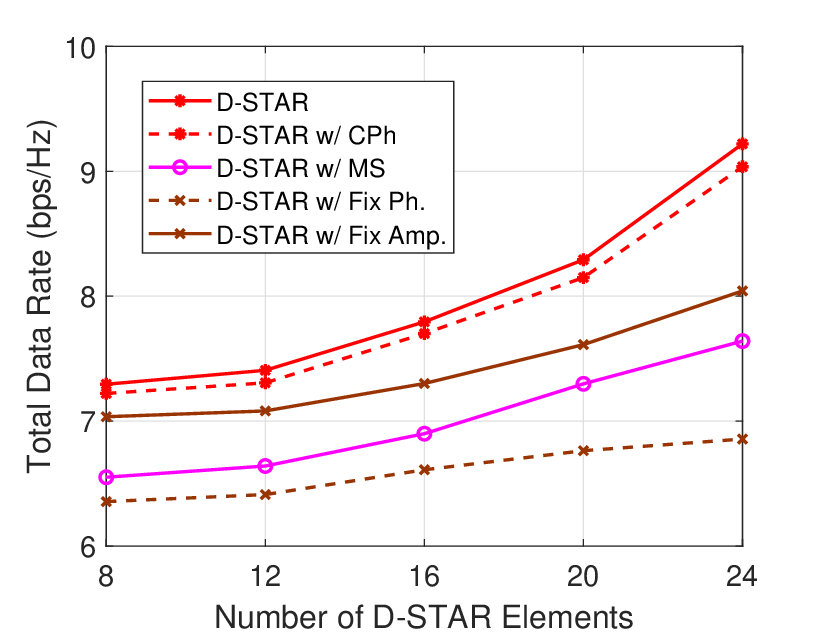} \label{fig:dios_RIS}}
		\caption{The performance of proposed D-STAR architecture w.r.t. different (a) numbers of BS transmit antennas with $P_t=30$ dBm and $M=8$ elements, (b) BS transmit power with $N_T=8$ and $M=8$ elements, and (c) D-STAR elements with $N_T=8$ and $P_t=30$ dBm. We compare D-STAR to its sub-schemes with CPh, MS, and optimization under fixed phase/amplitudes.}
		\label{fig:dios}
	\end{figure*}

In Fig. \ref{fig:split}, we partition the original D-STAR into several sub-STAR-RIS panels, while having the same total of $48$ elements. For example, splitting the D-STAR into two constituents means that a pair of STAR-P/-S each having $24$ elements is partitioned into two pairs, with each surface having $12$ elements. We observe that the partitioned surfaces are $10$ m away from the original one, as depicted in Fig. \ref{split}. We can observe from Fig. \ref{fig:ios_split} that all the throughput curves exhibit a concave shape, with the optimal point being at $8$ partitions. Upon increasing the number of partitions from $1$ to $8$, the rate improves thanks to the higher channel diversity, which allows the BS to focus its beamforming power on several beneficially selected STAR-RISs with better channel quality. However, increasing the number of partitions to $16$ reduces the throughput, because the more distant partitions suffer from a weak signal.

\subsection{Different Network Settings}

	In Fig. \ref{fig:dios}, we evaluate the performance of DBAP in D-STAR for different numbers of BS transmit antennas, transmit power, and D-STAR elements, as shown in Figs. \ref{fig:dios_ant}, \ref{fig:dios_power} and \ref{fig:dios_RIS}, respectively, which are compared to D-STAR with CPh, MS, and optimization under fixed phase/amplitudes. We can observe from Fig. \ref{fig:dios_ant} that more transmit antennas provide higher rate, since they can support higher directional beamforming gains as well as more beneficial alignment to D-STAR surfaces. They also offer higher spatial diversity for mitigating the SI and the inter-users interferences. Moreover, we can infer that D-STAR w/ CPh asymptotically approaches the performance of D-STAR. To elaborate a little further, MS has the lowest data rate when $N\geq 12$. This is because it will have a higher probability to misalign the BS antenna with D-STAR, hence leading to certain signal loss. In Fig. \ref{fig:dios_power}, it is observed that as expected, higher rate can be supported at a higher transmit power. However, MS outperforms the phase-only optimization, since having as few as $N_T=8$ antennas and 0-1 amplitude states for D-STAR suffers from a low rate even at a high transmit power. This means that when fewer antennas are deployed at the BS, D-STAR should be configured with each element activated in either for pure reflection or transmission. As shown in Fig. \ref{fig:dios_RIS}, more D-STAR elements proactively provide higher channel diversity, which potentially increases the received signal strength as well as mitigates the hostile interferences. To elaborate a little further, it reveals a greater improvement of around $2$ bps/Hz for D-STAR and D-STAR w/ CPh, when using $M=8$ to $M=24$ than the other benchmarks. Only a slight improvement is observed in comparison to other methods, with an increase of approximately $1$ bps/Hz for amplitude-only optimization and MS, as well as for phase-only optimization. 
	
\begin{figure*}[!t]
		\centering
		\subfigure[]
		{\includegraphics[width=2.1 in]{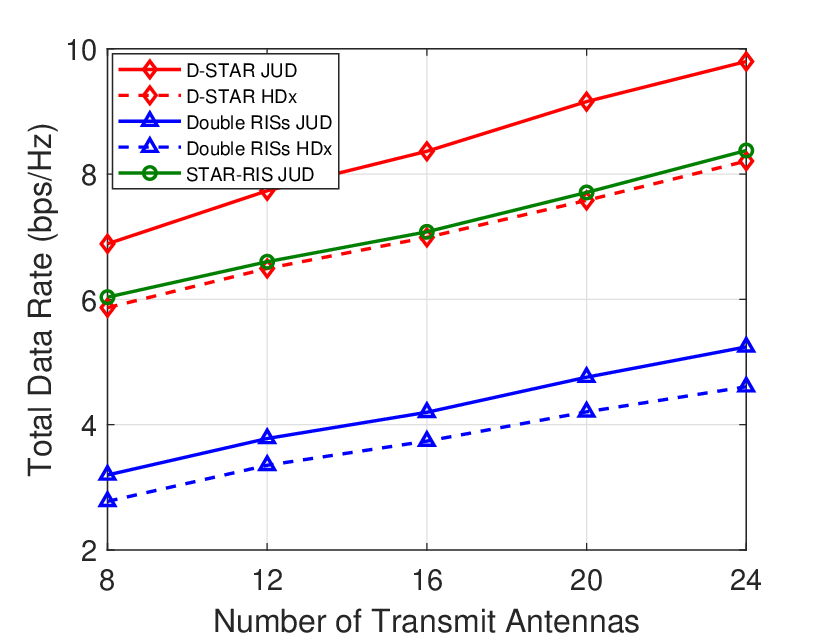} \label{fig:arch_ant}}
		\subfigure[]
		{\includegraphics[width=2.1 in]{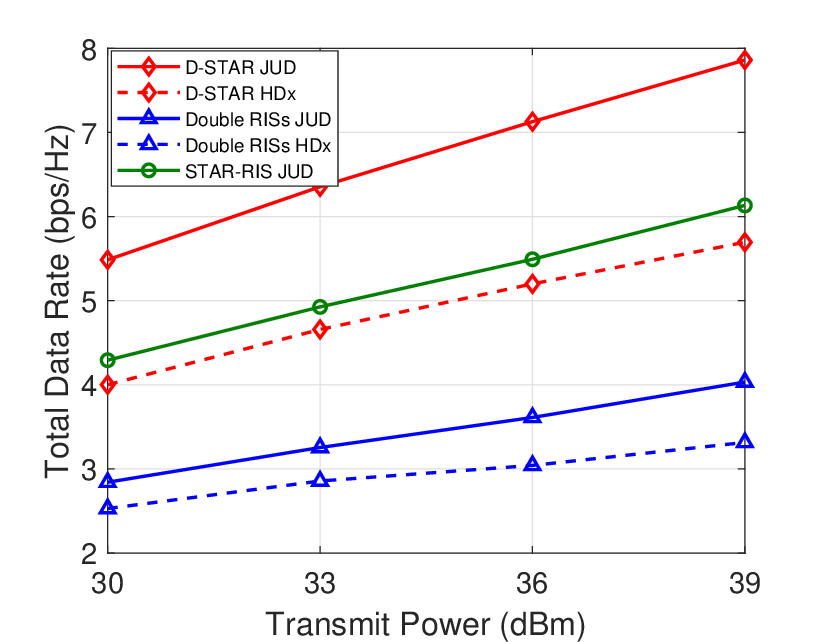} \label{fig:arch_power}}
		\subfigure[]
		{\includegraphics[width=2.1 in]{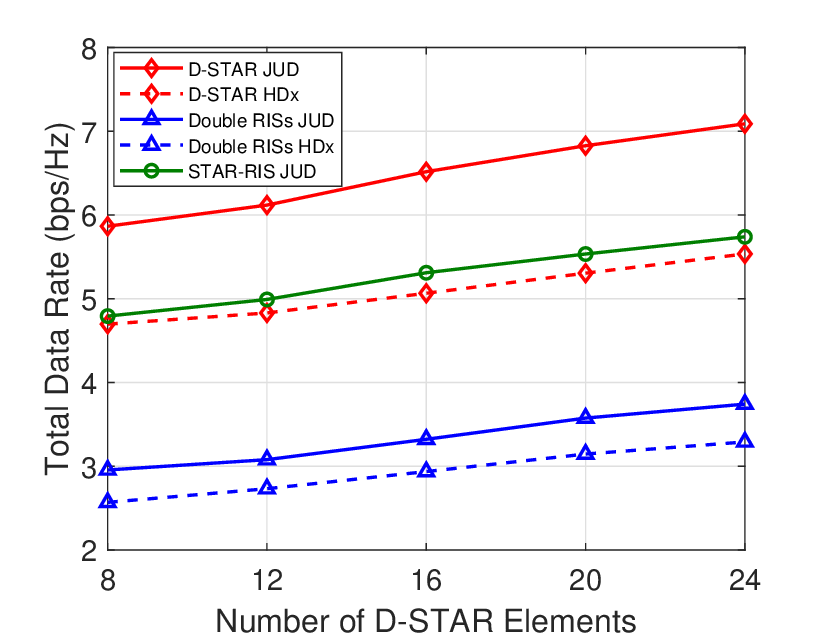} \label{fig:arch_RIS}}
		\caption{The performance of proposed D-STAR architecture w.r.t. different (a) numbers of BS transmit antennas, (b) BS transmit power, and (c) D-STAR elements. We compare different architectures and transmission techniques, i.e., D-STAR for HDx, single STAR-RIS for JUD, double RISs for JUD/HDx.}
		\label{fig:arch}
	\end{figure*}

\begin{figure*}[!t]
		\centering
		\subfigure[]
		{\includegraphics[width=2.1 in]{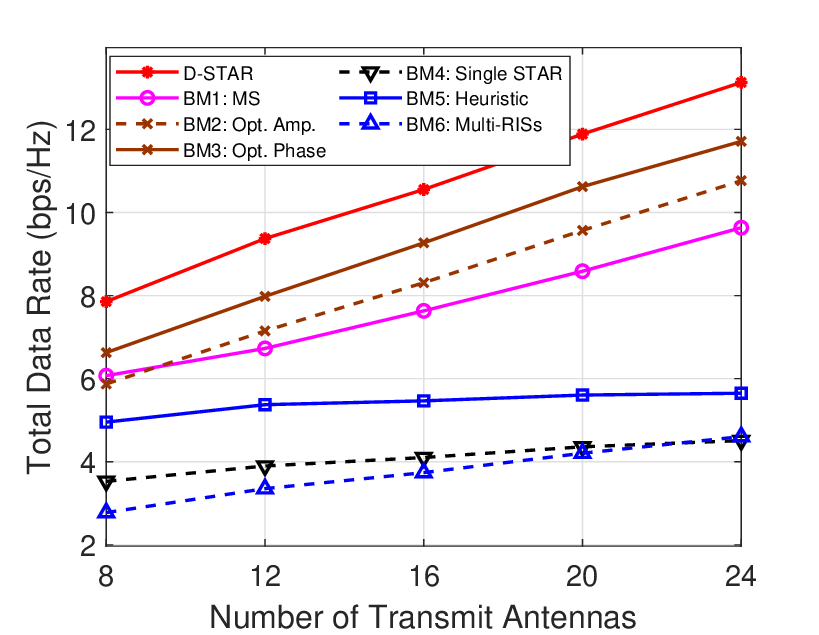} \label{fig:bm_ant}}
		\subfigure[]
		{\includegraphics[width=2.1 in]{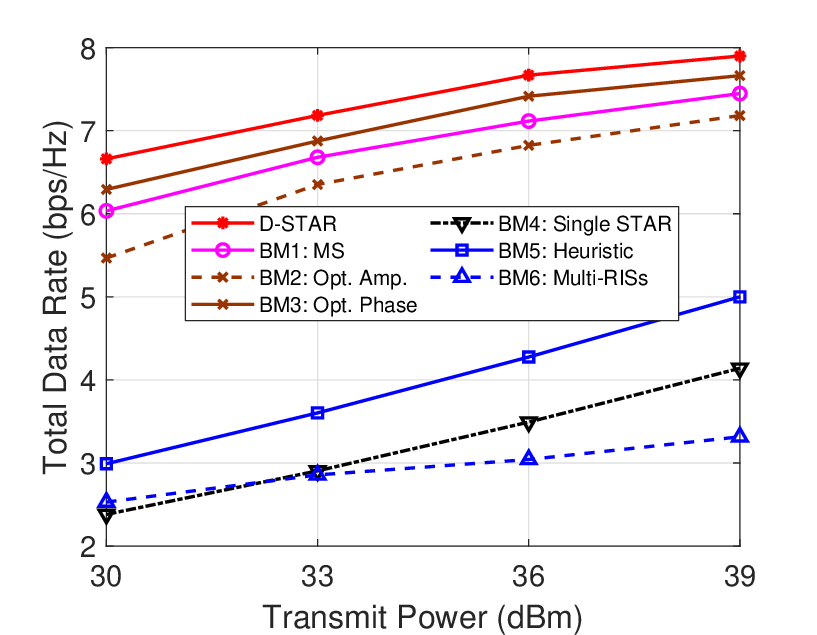} \label{fig:bm_power}}
		\subfigure[]
		{\includegraphics[width=2.1 in]{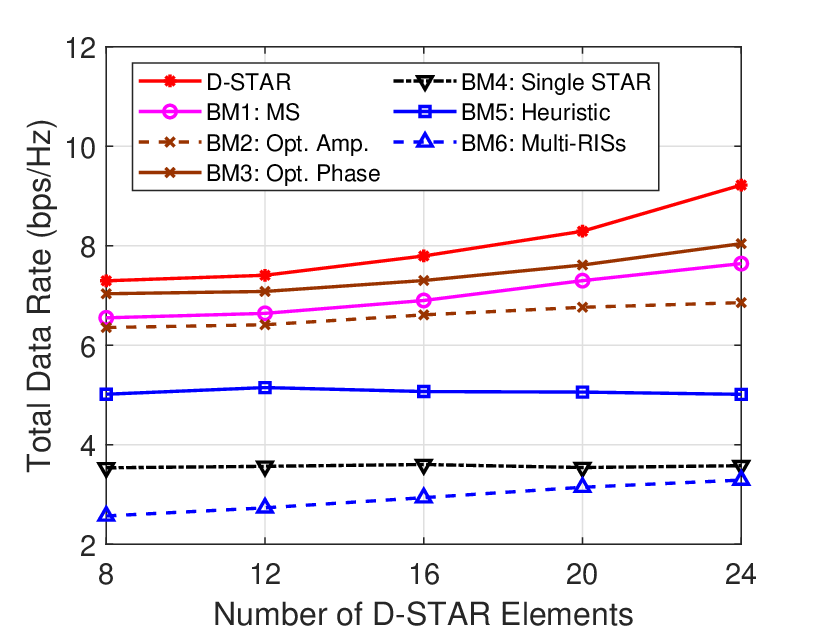} \label{fig:bm_RIS}}
		\caption{The performance comparison of proposed D-STAR with benchmarks of MS, single configuration optimization, single STAR-RIS, heuristic method and Multi-RISs in terms of different (a) numbers of BS transmit antennas, (b) BS transmit power, and (c) D-STAR elements.}
		\label{fig:bm}
	\end{figure*}

\subsection{Different Transmission Techniques}

In Fig. \ref{fig:arch}, we compare the rate of D-STAR in JUD to different architectures and transmission techniques, including D-STAR in HDx, double-RIS deployment in JUD/HDx, and a single STAR-RIS in JUD. Note that for comparing the performance in JUD and HDx under fair conditions, we take into account that HDx requires two time-slots, one for DL and one for UL transmission. Moreover, for double RISs the transmission sides of D-STAR are turned off, leaving the reflection function activated. This means that the UL/DL users in the S-region can only transmit/receive signals directly to/from the distant BS at a poorer signal quality than through STAR-RISs. The scenario of a single STAR-RIS is supposed to allow only STAR-P to be operated, whilst STAR-S is completely turned off. Again, we can have a higher rate, when more antennas, higher ower, or more D-STAR elements are available, as shown in Figs. \ref{fig:arch_ant}, \ref{fig:arch_power}, and \ref{fig:arch_RIS}, respectively. Without transmission function of STAR-RIS, we can observe that the double-RISs exhibit the worst rate, since the BS transmitter and the SU users directly transmit their signals to SD users and to the BS receiver, respectively under poor channel quality. The RIS function of the STAR-S reflection part is less useful in both the DL and UL. However, when JUD is considered, it is beneficial to alleviate the SU inter-user interference imposed on the SD users, resulting in a $12.3\%$ to $18.7\%$ rate improvement for the double RISs HDx to JUD link. By contrast, as a benefit of the full-coverage STAR function of D-STAR, we can enhance the rate by about $21\%$ to $37.2\%$.


\subsection{Benchmark Comparison}

In Fig. \ref{fig:bm}, we compare DBAP using D-STAR to several benchmarks (BM) found in the open literature for different numbers of BS transmit antennas, BS transmit power, and different number of D-STAR elements respectively, in Figs. \ref{fig:bm_ant}, \ref{fig:bm_power}, and \ref{fig:bm_RIS}. \textbf{BM1: MS} \cite{bm1} adopts mode switching for multiple non-orthogonal downlink users, where each element is configured either for reflection or transmission. \textbf{BM2: Opt. Amp.} only optimizes the amplitudes of D-STAR, leaving the phase shifts randomly configured. \textbf{BM3: Opt. Phase} \cite{couple2} employs relaxed transmission/reflection coefficient optimization and algebraic manipulations, aiming for solving the problem of coupled phase shifts. \textbf{BM4: Single STAR-RIS} \cite{fd_ok} considers the scenario of single STAR-RIS for JUD users. Recall that the direct BS to user link is are supposed to be blocked. \textbf{BM5: Heuristic} \cite{bm5} utilizes genetic algorithm based resource allocation. \textbf{BM6: Multi-RISs} \cite{bm6} employs only reflection functions in multi-RIS-based transmission.

Observe from Fig. \ref{fig:bm} that the worst performance is attained when multi-RISs are adopted, because the deployment and orientation are not optimized in support of all users. BM6 only outperforms the single STAR-RIS scenario when $P_t=30$ dBm since the low-powered signals impinging on STAR-RIS lead to insufficient separation between the transmission and reflection parts. Again, the direct links are unavailable, which further reduces the received signal power. To make the conventional single-sided operation of RISs realistic in BM6, additional geometric deployment and orientation problems should be considered. However, this may require more RISs and more RIS elements to achieve the same performance as D-STAR for full-plane service coverage. As a benefit of optimization, BM5 relying on a genetic algorithm has a higher rate than that of BMs 4 and 6. With 360-degree full-plane service coverage, the proposed D-STAR optimizing the BS beamforming, as well as the amplitudes and phase shifts attains the highest data rate in all cases compared to the state-of-the-art.

\section{Conclusions}\label{CON}

We have conceived the new D-STAR architecture for full-plane service coverage, with STAR-P of Fig. \ref{fig:sep} tackling the P-region interferences and with STAR-S alleviating the S-region ones. The non-linear and non-convex DL sum rate-maximization problem formulated is solved by alternating optimization by relying on the decomposed convex sub-problems of the BS beamformer and D-STAR configurations w.r.t. the amplitude and phase shifts. We proposed a DBAP optimization scheme for solving the respective sub-problems by the Lagrange dual with Dinkelbach's transformation, ADMM with SCA, and PCCP. Our simulation results have characterized the optimal inter-D-STAR distances and partitioning. They also revealed that the proposed D-STAR architecture outperforms the conventional single RIS, single STAR-RIS, and HDx networks. Furthermore, the proposed DBAP in D-STAR achieves the highest throughput amongst the state-of-the-art solutions in the open literature.

\appendix
\section{Appendix}
\subsection{Proof of Lemma \ref{lemma_LD}} \label{Appendix_LD}

The original Lagrangian dual transform is based on the weighted sum-of-logarithms problem having a form of $\max_{\boldsymbol{\Xi}, \gamma_{u,k}}\ f_r(\boldsymbol{\Xi}, \gamma_{u,k})=\sum_{k\in\mathcal{K}_u} w_k \log_2 \left(1+\frac{A_k(\boldsymbol{\Xi})}{B_k(\boldsymbol{\Xi})}\right)$, where $w_k$ is the weight of each logarithmic expression. Considering an equal unit weight of $w_k=1$ yields the same problem. When the number of iterations tends to infinity, the auxiliary variable will asymptotically approach the original fractional parameter, i.e., $\gamma_{u,k}^* = \frac{A_k(\boldsymbol{\Xi}^*)}{B_k(\boldsymbol{\Xi}^*)}$. Accordingly, the following equation holds:
	\begingroup
	\allowdisplaybreaks
\begin{align}
	& - \gamma_{u,k}^* + \frac{(1+\gamma_{u,k}^*) A_k(\boldsymbol{\Xi}^*) }{A_k(\boldsymbol{\Xi}^*)+ B_k(\boldsymbol{\Xi}^*)} \notag \\
	& = -\frac{A_k(\boldsymbol{\Xi}^*)}{B_k(\boldsymbol{\Xi}^*)} + \frac{ \left( 1+\frac{A_k(\boldsymbol{\Xi}^*)}{B_k(\boldsymbol{\Xi}^*)}\right) A_k(\boldsymbol{\Xi}^*)}{A_k(\boldsymbol{\Xi}^*) + B_k(\boldsymbol{\Xi}^*)} \notag \\
	& =  -\frac{A_k(\boldsymbol{\Xi}^*)}{B_k(\boldsymbol{\Xi}^*)} + \frac{A_k(\boldsymbol{\Xi}^*)}{B_k(\boldsymbol{\Xi}^*)} = 0.
\end{align}
\endgroup
Therefore, we can obtain the additional dual term of $- \sum_{k\in\mathcal{K}_u} \gamma_{u,k} + \sum_{k\in\mathcal{K}_u} \frac{(1+\gamma_{u,k}) A_k(\boldsymbol{\Xi}) }{A_k(\boldsymbol{\Xi}) + B_k(\boldsymbol{\Xi})}$. Considering this additional term in the original problem yields $\eqref{LD}$. This completes the proof. \hfill\(\blacksquare\)

\subsection{Proof of Lemma \ref{lemma_DIN}} \label{Appendix_DIN}

We know that there exists an optimal value for the fractional programming, which is given by
	\begingroup
	\allowdisplaybreaks
\begin{align}
	 &\mathop{\max}_{\boldsymbol{\Xi}} \  \Lambda(\boldsymbol{\Xi}) 
	= \mathop{\max}_{\boldsymbol{\Xi}} \  \sum_{k\in \mathcal{K}_u} \frac{A_k(\boldsymbol{\Xi})}{C_k(\boldsymbol{\Xi})} \notag\\
	&\Leftrightarrow \mathop{\max}_{\boldsymbol{\Xi}, \forall k\in\mathcal{K}_u} \  \frac{A_k(\boldsymbol{\Xi})}{C_k(\boldsymbol{\Xi})} 
	= \frac{A_k(\boldsymbol{\Xi}^*)}{C_k(\boldsymbol{\Xi}^*)} 
	 \triangleq \lambda^{*},
\end{align}
\endgroup
where $C_k(\boldsymbol{\Xi}) = A_k(\boldsymbol{\Xi})+ B_k(\boldsymbol{\Xi})$.
Therefore, the following problem holds
\begin{align} \label{Din2}
	& \mathop{\max}_{\boldsymbol{\Xi}} \ A_k(\boldsymbol{\Xi}) - \lambda^* C_k(\boldsymbol{\Xi})= A_k(\boldsymbol{\Xi}^*) - \lambda^* C_k(\boldsymbol{\Xi}^*) = 0.
\end{align}
Replacing the fractional term in $\eqref{LD}$ by $A_k(\boldsymbol{\Xi}) - \lambda C_k(\boldsymbol{\Xi})$ yields $\eqref{Din}$, which holds when the number of iterations tends to infinity. Moreover, we know from $\eqref{LD}$ that the first two terms of $\sum_{k\in\mathcal{K}_u} \log_2 (1+\gamma_{u,k})$ and $\sum_{k\in\mathcal{K}_u} \gamma_{u,k}$ are regarded as constants acquired from the previous outcomes without any variables to be determined. Therefore, the optimization of $\eqref{LD}$ is equivalent to that utilizing $\eqref{Din}$. This completes the proof.\hfill\(\blacksquare\)

\subsection{Proof of Lemma \ref{lemma_sca}} \label{Appendix_sca}

We adopt Taylor expansion as $f^-(x)= f^-(x_0) + \nabla_x^{\mathcal{H}} f^-(x_0)(x-x_0) + \nabla_x^{2\, \mathcal{H}} f^-(x_0)\frac{(x-x_0)^2}{2} +\cdots$. Let us now assume that $O(x^n) = \nabla_x^{n\, \mathcal{H}} f^-(x_0)\frac{(x-x_0)^n}{n!}$ denotes the term having derivatives higher than the second order. We can then have $f^-(x)=f^-(x_0) + \nabla_x^{\mathcal{H}} f^-(x_0)(x-x_0) + \lim_{N\rightarrow \infty} \sum_{n=2}^{N} O(x^n)$. Without $O(x^n)$, we can acquire a lower bound of the affine function, i.e., $\tilde{f}^- (x) = f^-(x_0) + \nabla_x^{\mathcal{H}} f^-(x_0)(x-x_0)$. When we neglect the term with $n\geq 2$ in $f^-(x)$, we can derive the concave objective of $f(x)\geq f_{cav}(x)$, where $f_{cav}(x) = f^+(x) + f^-(x_0) + \nabla_x^{\mathcal{H}} f^-(x_0)(x-x_0)$. Moreover, considering that $O(x^n)$ is comparatively small and tends to zero, we can have $f(x) \approx f^+(x) + f_{cav}(x)$. This completes the proof. \hfill\(\blacksquare\)

\subsection{Proof of Lemma \ref{lemma_sym}}
\label{Appendix_sym}

We adopt mathematical induction to prove this lemma by commencing with $N_1=3$ and $N_2=2$, i.e., $\mathbf{d}=\left[ d_1, d_2, d_3 \right]^{\rm T}$, $\mathbf{D}= [D_{1,1},D_{1,2}; D_{2,1}, D_{2,2}; D_{3,1}, D_{3,2}]$, $\boldsymbol{\phi}=\left[ \phi_1, \phi_2, \phi_3\right]^{\rm T}$, and $\mathbf{w}=\left[w_1, w_2\right]^{\rm T}$. The left-hand side of $\eqref{dd}$ becomes 
\begin{align} \label{c_matrix}
	 & \left( d_1 \phi_1 D_{1,1} + d_2 \phi_2 D_{2,1} + d_3 \phi_3 D_{3,1 } \right) w_1 \notag \\
	 &\qquad + \left( d_1 \phi_1 D_{1,2} + d_2 \phi_2 D_{2,2} + d_3 \phi_3 D_{3,2} \right) w_2 \notag\\
	 &= [w_1, w_2] \begin{bmatrix}
d_1 D_{1,1} & d_2 D_{2,1} & d_3 D_{3,1}\\
d_1 D_{1,2} & d_2 D_{2,2} & d_3 D_{3,2}
\end{bmatrix}
	\left[ \phi_1, \phi_2, \phi_3 \right]^{\rm T} \notag\\
	&= [w_1, w_2] \begin{bmatrix}
d_1 & d_2  & d_3 \\
d_1 & d_2  & d_3 
\end{bmatrix} \odot
	\begin{bmatrix}
 D_{1,1} &  D_{2,1} & D_{3,1}\\
 D_{1,2} &  D_{2,2} & D_{3,2}
\end{bmatrix}
	\left[ \phi_1, \phi_2, \phi_3 \right]^{\rm T}
\end{align}
which represents the exact form at the right-hand side of $\eqref{dd}$. Accordingly, upon considering arbitrary numbers for $N_1$ and $N_2$, we can proceed further from $\eqref{c_matrix}$ to acquire a general expression as
\begin{align} \label{dd2}
	[w_1, \cdots , w_{N_2}]
	\begin{bmatrix}
d_1 D_{1,1}  & \cdots & d_{N_1} D_{N_1,1}\\
 \vdots  & \ddots & \vdots \\
d_1 D_{1,N_2} & \cdots & d_{N_1} D_{N_1,N_2}
\end{bmatrix}
	\left[ \phi_1, \cdots , \phi_{N_1} \right]^{\rm T},
\end{align}	
which is the same as the outcome associated with the parameters of $N_1+1$ and $N_2+1$. Both have results that are identical to the right-hand side of $\eqref{dd}$. This completes the proof. \hfill\(\blacksquare\)

\subsection{Proof of Corollary \ref{Cor_dd}}
\label{Appendix_dd}
We can infer that $\eqref{dd2}$ is performed for a row vector in $\mathbf{U}$. Therefore, we may carry out the matrix operations in  $\eqref{dd2}$ to create $N_{3}$ independent diagonal blocks. To evaluate $\eqref{dd}$ in each block for different $\mathbf{U}_n,\forall 1\leq n \leq N_3$, the Kronecker product is required for $\mathbf{w}^{\rm T}$ and $\mathbf{D}^{\rm T}$ in order to prevent non-zero values for the non-diagonal elements. Similarly, ${\rm rep}$ is executed for each row vector $\mathbf{U}_{(n,:)}$. Leveraging the above operations yields $\eqref{DD}$. This completes the proof. \hfill\(\blacksquare\)

\subsection{ADMM Fundamentals} \label{Appendix_ADMM}
Given the convex functions of $f(x)$ and $g(z)$ and convex sets of $\mathcal{X}$ and $\mathcal{Z}$, the associated ADMM problem is formulated as
\begin{subequations}
\begin{align}
\mathop{\min}\limits_{x\in\mathcal{X}, z\in \mathcal{Z}} & \quad  f(x)+g(z) \\
    \text{s.t.} &\quad  Ax+Bz=c.
\end{align}
\end{subequations}
The augmented Lagrangian is acquired as
\begin{align}
& \mathcal{L}_{\rho}(x,z,y) = \notag \\
& f(x) \!+\! g(z) \!+\! y^{\rm T}(Ax+Bz-c) \!+\! \frac{\rho}{2} \lVert Ax+Bz-c \rVert^2.
\end{align}
Therefore, the respective optimizations and dual update are acquired by
\begingroup
 \allowdisplaybreaks
\begin{align}
	x^{(t+1)} & \leftarrow \argmin_{x\in\mathcal{X}} \ \mathcal{L}_{\rho}\left( x,z^{(t)},y^{(t)} \right),\\
	z^{(t+1)} & \leftarrow \argmin_{z\in\mathcal{Z}} \  \mathcal{L}_{\rho}\left( x^{(t+1)},z,y^{(t)} \right),\\
	y^{(t+1)} & \leftarrow y^{(t)} + \rho \left( Ax^{(t+1)} + Bz^{(t+1)} - c \right).
\end{align}
\endgroup
Optimality is achieved when \textit{primal} and \textit{dual} feasibility are respectively achieved as $Ax+Bz-c=0$ and $\nabla f(x)+A^{\rm T} y=0$, $\nabla g(z)+B^{\rm T} y=0$. Since $z^{(t+1)}$ minimizes $\mathcal{L}_{\rho}(x^{(t+1)}, z, y^{(t)})$, we have
\begin{align*}
	0 &= \nabla g(z^{(t+1)}) + B^{\rm T} y^{(t)} + \rho B^{\rm T} (Ax^{(t+1)}+B z^{(t+1)} -c) \notag\\
	&= \nabla g(z^{(t+1)}) + B^{\rm T} y^{(t+1)}.
\end{align*}
Accordingly, the ADMM dual variables of $\{ x^{(t+1)}, z^{(t+1)}, y^{(t+1)}\}$ satisfy the second dual feasibility condition. Therefore, primal and dual feasibility are achieved as $t \rightarrow \infty$.

\bibliographystyle{IEEEtran}
\bibliography{IEEEabrv}

\end{document}